\documentclass[11pt]{amsart} 
\usepackage[lmargin=1in,rmargin=1in,tmargin=1in,bmargin=1in]{geometry}
\usepackage[ps,all,arc,rotate]{xy}
\usepackage{graphicx, float, epstopdf}
\usepackage{hyperref, color}
\usepackage{centernot}
\usepackage{fancyhdr}
\usepackage[utf8]{inputenc}
\usepackage{amsfonts,amssymb,amsmath,amsthm,mathrsfs,longtable}
\usepackage{graphics, setspace}
\usepackage{braket}
\usepackage{mathtools}
\usepackage{tikz}

\numberwithin{equation}{section}
\numberwithin{figure}{section}
\allowdisplaybreaks[4]         
 
\newtheorem{lemma}{Lemma}[section]
\newtheorem{theorem}{Theorem}[section]
\newtheorem{proposition}{Proposition}[section]
\newtheorem{example}{Example}[section]

\newtheorem{corollary}[lemma]{Corollary}
\theoremstyle{definition}
\newtheorem{definition}{Definition}[section]
\newtheorem{remark}{Remark}[section]
\newcommand{\R}{\mathbb{R}}

\newcommand{\outdeg}{\operatorname{outdeg}}

\newcommand{\sumtwo}{\operatorname*{\sum\sum}}

\newcommand{\modu}{\operatorname{mod}}

\newcommand{\norm}[1]{\left\lVert#1\right\rVert}

\newcommand{\mycomment}[1]{}

\newcommand{\rr}{\m{R}}

\newcommand{\la}{\lambda}

\newcommand{\m}[1]{\mathbb{#1}}
\newcommand{\mc}[1]{\mathcal{#1}}

\newcommand{\ip}[1]{\langle #1 \rangle}

\begin{document}

\title[Systemic importance with missing data]{Estimating systemic importance with missing data in input-output graphs}

\author[Geneson]{Jesse Geneson}
\address{RAND Corporation, Engineering and Applied Sciences, Pittsburgh, PA}
\email{\texttt{jgeneson@rand.org}}

\author[Moon]{Alvin Moon}
\address{RAND Corporation, Engineering and Applied Sciences, Santa Monica, CA}
\email{\texttt{alvinm@rand.org}}

\author[Robles]{Nicolas Robles}
\address{RAND Corporation, Engineering and Applied Sciences, Arlington, VA}
\email{\texttt{nrobles@rand.org}}

\author[Strong]{Aaron Strong}
\address{RAND Corporation, Economics, Sociology \& Statistics, Santa Monica, CA}
\email{\texttt{astrong@rand.org}}

\author[Welburn]{Jonathan Welburn}
\address{RAND Corporation, Engineering and Applied Sciences, Santa Monica, CA}
\email{\texttt{jwelburn@rand.org}}

\subjclass[2020]{Primary: 15A09, 15A29. Secondary: 91B60, 91B74, 60C05. \\ \indent \textit{Keywords and phrases}: Input-output graphs, systemic importance, Cobb-Douglas productivity model, Leontief inverse, production network, influence vector, sharp bounds, directed chains, \texttt{PageRank}.}

\maketitle

\vspace{-10pt}

\begin{abstract}
    In the context of the Cobb-Douglas productivity model we consider the $N \times N$ input-output linkage matrix $W$ for a network of $N$ firms $f_1, f_2, \cdots, f_N$. The associated influence vector $v_w$ of $W$ is defined in terms of the Leontief inverse $L_W$ of $W$ as 
    \begin{align*}
    v_W = \frac{\alpha}{N} L_W \vec{\mathbf{1}} \quad \textnormal{where} \quad L_W = (I - (1-\alpha) W')^{-1},
    \end{align*}
    $W'$ denotes the transpose of $W$ and $I$ is the identity matrix. Here $\vec{\mathbf{1}}$ is the $N \times 1$ vector whose entries are all one. The influence vector is a metric of the importance for the firms in the production network. Under the realistic assumption that the data to compute the influence vector is incomplete, we prove bounds on the worst-case error for the influence vector that are sharp up to a constant factor. We also consider the situation where the missing data is binomially distributed and contextualize the bound on the influence vector accordingly. We also investigate how far off the influence vector can be when we
only have data on nodes and connections that are within distance $k$ of some source node.
    A comparison of our results is juxtaposed against \texttt{PageRank} analogues. We close with a discussion on a possible extension beyond Cobb-Douglas to the Constant Elasticity of Substitution model, as well as the possibility of considering other probability distributions for missing data.
\end{abstract}


\section{Introduction}
\subsection{Motivation from economics and policy} A growing literature has sought to understand the macro economy and the potential for systemic risk at the firm level through models of interfirm linkages.  However, due to fundamental challenges in gathering data on interfirm linkages and associated values, no complete dataset on global interfirm networks exists (Carvalho and Tahbaz-Salehi \cite{premier}).
For example, most countries do not require full disclosure of the set of customers and suppliers that a firm may have. If an important firm or linkage is missing, then this may have consequences when considering systemic risk or systemic importance of any single firm or connection. Barrot and Sauvagnat \cite{barrotSauvagnat2016} and Welburn \textit{et al.} \cite{welburn2020} have exploited Financial Accounting Standard 131, which states that firms must report suppliers that make up more than 10\% of costs. This standard only applies to publicly traded firms and provides a lower bound on what they must report. Some have looked to more complete country level data offered in Japan (e.g., Arata \cite{arata}, Carvalho, Nirei \textit{et al.} \cite{earthquake}), while some have utilized expanded datasets of firm level connections through proprietary vendors such as FactSet or Bloomberg (e.g. Wu and Birge \cite{wubirge}, Crosignani, Macchiavelli \textit{et al.} \cite{fed2021}).  Others have sought to address the challenge of missing data in the aforementioned datasets through inference techniques (e.g., Welburn, Strong \textit{et al.} \cite{welburn2020, welburn2023}). However, uncertainty still remains.

Sectoral relationships or country level trade networks may be relatively more complete but do not describe the disaggregated nature of firm level connections that may drive disruptions. The aim of much of this literature is to understand how disruptions of one firm are likely to affect the system. In a sector level approach, these disruptions may be mitigated simply due to the scale of the economy at that level and that individual firm disruptions may be small relative to the economy or sector. Thus, we may lose fidelity as to the consequences of a disruption that may percolate across an economy at the firm rather than sector level.

Economic models that focus on the firm level may provide insights that cannot be gained at the sector level but at the expense of calibration data that matches reality. Economic models that focus on the sector level may mitigate some of the realities of disruptions while having more representative interactions across the economy. This work aims to understand how large difference errors in calibration may percolate to systemic risks and measures of systemic importance. That is, we estimate error bounds on systemic importance as a function of the network and share of missing data. Our goal is ultimately to provide insights into how to model systemic risk, systemic importance, and the role of aggregate and disaggregated data. 

Knowing the sources of systemic risk is important in identifying potential resilience investment that can take place at the system rather than firm or sector level. By having a system view rather than a firm level view, there may be cost effective strategies that can be identified to reduce the impact of cascading disruptions that may not be readily identified by a firm level approach. Strategies that focus on the system may result in not only investments in security but redundancy, adaptation, and reducing disruption impacts rather than the probability that a disruption occurs or the direct impacts to a single firm.

\subsection{Motivation from economic theory}
Suppose we have an economy with a fixed number $n$ of competitive industries producing a distinct good. Such an economy is called static. The end result of each good can only be
\begin{itemize}
    \item either final consumption by households, 
    \item or intermediate output for other goods in a larger chain of interactions.
\end{itemize}
The most accepted model for such economies is the Cobb-Douglas model \cite{cobbdouglas}. In this model, each of the industries transforms intermediate inputs and labor into final products \cite[$\mathsection$2]{premier} and the underlying assumption is that of constant returns. To formalize this notion, suppose the output of entity $i$ can be factored into the product
\begin{align} \label{eq:cobbdouglas}
    y_i = z_i \zeta_i \ell_i^{\alpha_i} \prod_{j=1}^n x_{ij}^{a_{ij}},
\end{align}
where each factor is defined as follows:
\begin{enumerate}
    \item $\ell_i$ is the amount of labor hired by industry $i$,
    \item $x_{ij}$ is the quantity of good $j$ used for production of good $i$,
    \item $\alpha_i$ is the share of labor in industry $i$'s production technology (and hence $\alpha_i>0$),
    \item $z_i$ is a Hicks-neutral productivity shock \cite[$\mathsection$2.1]{premier},
    \item $\zeta_i = \zeta_i(\alpha_i, a_{ij})$ is a normalization constant. A benign choice for $\zeta_i$ is $\zeta_i = \alpha^{-\alpha_i} \prod_{j=1}^n a_{ij}^{-a_{ij}}$,
    \item the reliance factor $a_{ij}$ indicates that firms in an industry need to rely on goods produced by other industries as intermediate inputs for their own production. In all cases, we have $a_{ij} \ge 0$. If $a_{ij} = 0$, then good $j$ was not employed in the production of good $i$. If $a_{ij}$ is large, then good $j$ is a more important input when producing good $i$. Generally $a_{ij} \ne a_{ji}$ and $a_{ii}$ may be strictly greater than $0$.
\end{enumerate}
In order to have constant returns we must have the relation
\begin{align} \label{eq:constantreturn} 
    \alpha_i + \sum_{j=1}^n a_{ij} = 1 \quad \text{ for all } i \in \{1,2,\cdots,n\}.
\end{align}
The associated utility function is given by
\begin{align} \label{eq:utility} 
u(c_1, \cdots, c_n) = \sum_{i=1}^n \beta_i \log \frac{c_i}{\beta_i},
\end{align}
where $c_i$ is the amount of good $i$ consumed and $\beta_i \ge 0$ measures various goods' shares in the household's utility function, normalized such that $\sum_{i=1}^n \beta_i = 1$. The utility function \eqref{eq:utility} along with \eqref{eq:cobbdouglas} give a full description of the economic landscape.

The connection to macroeconomic phenomena is made through an industry's Domar weight. This is defined as that industry's sales as a fraction of the national GDP
\begin{align} \label{eq:domar}
    \lambda_i = \frac{p_i y_i}{\operatorname{GDP}},
\end{align}
where $p_i$ is the price of good $i$ and $y_i$ is industry $i$'s output. Profits are defined by
\begin{align}
    \pi_i = p_i y_i - w l_i  - \sum_{j=1}^n p_j x_{ij},
\end{align}
where $w$ is the wage. This implies that, up to first order, the conditions corresponding to firms in industry $i$ are given by $x_{ij} = a_{ij}p_iy_i/p_j$ and $l_i = \alpha_i p_i y_i/w$. If we now employ these in the Cobb-Douglas model then 
\begin{align}
    \log \frac{p_i}{w} = \sum_{j=1}^n a_{ij} \log \frac{p_j}{w} - \log z_i.
\end{align}
Typically one sets the notation $\varepsilon_i = \log z_i$ to denote the log-productivity shock to firms in industry $i$. This naturally leads to a system of equations and it is most easily handled in matrix notation as
\begin{align}
    \hat {\vec p} = {A} \hat {\vec p} - {\vec \varepsilon},
\end{align}
where ${A}$ is the economy's input-output matrix. This can be solved formally to yield
\begin{align} \label{eq:Leontieffinverse}
    \hat {\vec p} = - ({I} - {A})^{-1}{\vec \varepsilon}.
\end{align}
This is the equilibrium vector of log-relative prices. 

\subsection{Input-output linkage matrix, Leontief inverse, and influence vector}

Given a network of $N$ firms $f_1, f_2, \dots, f_N$, we define their \textit{input-output linkage matrix} $W \in M_N(\R)$ to be the $N \times N$ matrix for which entry $w_{i j}$ represents the share of intermediate inputs to firm $f_i$ which are from firm $f_j$. Specifically, if firm $f_i$ has some amount $\text{II}_i$ of intermediate input use, then the amount of intermediate input to firm $f_i$ from firm $f_j$ is equal to $w_{i j} \text{II}_i$. We assume that $w_{i i} = 0$ for all $i$. 

For any matrix $W$, we let $W'$ denote the transpose of $W$. Moreover, we let $I_n$ denote the $n \times n$ identity matrix, and we drop the subscript $n$ when the dimensions are clear from the context. Let $\alpha \in (0, 1)$ be the parameter representing the share of labor in the network of firms. We assume that $\sum_{j = 1}^{n} w_{i j} = 1$ for all $i = 1, \dots, n$. Thus, $\sum_{i = 1}^{n} w'_{i j} = 1$ for all $j = 1, \dots, n$.

\begin{definition} \label{def:leontief}
Suppose that $\alpha \in (0,1)$. The \textit{Leontief inverse} $L_{W}$ of $W \in M_N(\rr)$ is 
\begin{align}
L_W = (I - (1-\alpha) W')^{-1}
\end{align}
whenever $(I-(1-\alpha) W')$ is invertible. The \textit{influence vector} of $W$ is 
\begin{align}
v_W = \frac{\alpha}{N} L_W \vec{\mathbf{1}},
\end{align} 
where $\vec{\mathbf{1}}$ is the $N\times 1$ vector whose entries are all $1$s. 
\end{definition} 

The $(i,j)$ entry of $L_W$, denoted by $\ell_{ij}$, measures the importance of industry $j$ as a direct and indirect supplier to industry $i$ in the economy \cite[p. 639]{premier}. The Leontief inverse also encapsulates how sectoral productivity shocks propagate downstream to other sectors through the input-output matrix. Thus, the aggregate output depends on the intersectoral network of the economy through the Leontief inverse \cite[p. 1985]{acemoglu}.

The influence vector is a measure of importance for the firms in the production network, i.e., firm $i$ is more important than firm $j$ if $v_{W,i} > v_{W,j}$. In the context of networks of websites, the influence vector is the same as the \texttt{PageRank} vector used in Google search rankings \cite{pagerank}. 

\begin{example}
Let us consider the following graph $G$
\begin{figure}[h]
    \centering
    \includegraphics[scale=1]{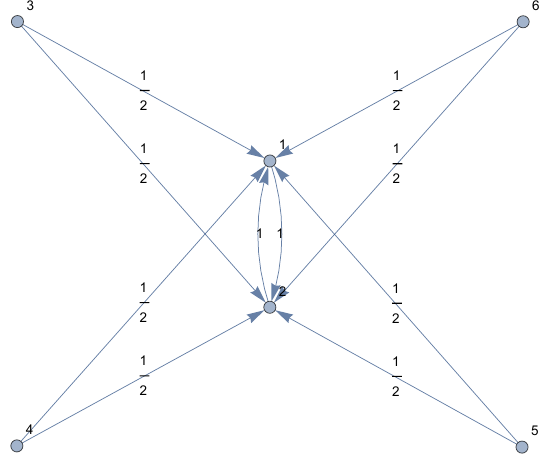}
    \caption{Graph $G$ with adjacency matrix \eqref{eq:adjancecymatrixG}}
    \label{fig:enter-label}
\end{figure}
with adjacency matrix
\begin{align} \label{eq:adjancecymatrixG}
G_{ij} =
\left( {\begin{array}{*{20}{c}}
  0&1&0&0&0&0 \\ 
  1&0&0&0&0&0 \\ 
  1&1&0&0&0&0 \\ 
  1&1&0&0&0&0 \\ 
  1&1&0&0&0&0 \\ 
  1&1&0&0&0&0 
\end{array}} \right),
\end{align}
and weights 
\begin{align}
    w_{12}&=w_{21}=1, \nonumber \\
    w_{3,1}&=w_{3,2}=w_{4,1}=w_{4,2}=w_{5,1}=w_{5,2}=w_{6,1}=w_{6,2}=\frac{1}{2}.
\end{align}
In this case the influence vector is given by $$\bigg(\frac{1}{2} - \frac{\alpha}{3}, \frac{1}{2} - \frac{\alpha}{3}, \frac{\alpha}{6}, \frac{\alpha}{6}, \frac{\alpha}{6}, \frac{\alpha}{6}\bigg).$$
\end{example}

\subsection{Past results}

Acemoglu \textit{et al.} \cite{acemoglu} investigated the influence vector and other measures of importance in the context of networks of economic sectors. They showed that when $W$ is the input-output linkage matrix for the sectors of the economy, the $i^{\text{th}}$ coordinate of the influence vector $v_W$ is equal to the equilibrium share of sales of sector $i$. They also showed that the aggregate volatility of the network scales with the Euclidean norm of the influence vector.

A number of error bounds have been proved for the influence vector in the context of \texttt{PageRank}. Ipsen and Wills \cite{ipsen} proved an exponential bound on the convergence rate of the \texttt{PageRank} vector when it is computed using the power method. They also presented error bounds on the influence vector when the personalization vector (the all-ones vector in the formula for the influence vector) is perturbed, when $\alpha$ is perturbed, and when the transition matrix is perturbed. In particular, we use the following bound from their paper.

\begin{theorem}\label{thm:ml_pr}\cite{ipsen}
    If $W$ and $U$ are input-output matrices with $W = U+B$ and $\alpha \in (0, 1)$, then \[\norm{v_U-v_W}_1 \le \frac{(1-\alpha)\norm{B}_{\infty}}{\alpha}.\] 
\end{theorem}

A number of variants of \texttt{PageRank} have been investigated in the literature. Bahmani \textit{et al.} \cite{bahmani} analyzed the \texttt{PageRank} vector of a dynamic network. They developed an algorithm for estimating the \texttt{PageRank} vector at each time step, and they showed that their estimate is close in expectation to the true \texttt{PageRank} vector. 

Gleich \textit{et al.} \cite{gleich} introduced the multilinear \texttt{PageRank} vector, which generalizes the original \texttt{PageRank} vector and does not always have a unique solution. They developed convergence bounds for various methods of computing the multilinear \texttt{PageRank} vector, and they also determined parameter regimes for which the multilinear \texttt{PageRank} vector has a unique solution. Li \textit{et al.} \cite{ml_pagerank} found additional parameter regimes where the multilinear \texttt{PageRank} vector is uniquely determined, and they also proved perturbation bounds for the multilinear \texttt{PageRank} vector within those parameter regimes. 

\subsection{Our results} 

In this paper, we prove bounds on the error in the influence vector in the presence of missing data. In particular, we show that if at most a $\delta$ share of the data for the intermediate input use of each firm is missing, then the worst-case error in the influence vector is on the order of $\delta$. We obtain similar bounds for networks of firms and goods.

We also investigate error bounds for the influence vector when there is only data for a specified neighborhood of the network. In particular, we obtain error bounds for chain-like networks of firms when data is only available for an initial segment of the chain. Furthermore, we discuss applications of our error bounds to \texttt{PageRank}.

We introduce notation and definitions in Section~\ref{s:def} that are used throughout the paper. In Section~\ref{s:error_bounds}, we prove an upper bound of $O(\delta)$ on the worst-case error for the influence vector when at most a $\delta$ share of the data for the intermediate input use of each firm is missing. We prove that the upper bound is sharp up to a constant factor in the same section. In order to do so, we construct for each $n > 2$ a network of $n$ firms that is missing data for at most a $\delta$ share of the intermediate input use of each firm, and we show that the error in the influence vector for the constructed network is $\Omega(\delta)$.

In Section~\ref{s:extremal}, we derive sharp extremal bounds on the coefficients of the influence vector. In particular, we exactly determine the maximum and minimum possible coefficients of the influence vector in a network of $n$ firms. In Section~\ref{s:negative}, we consider a scenario where the missing data accounts for at most a $\delta$ share of the firms in the network, rather than a $\delta$ share of the intermediate input use of each firm. We show that the worst-case error in this scenario does not go to $0$ as $\delta$ goes to $0$. In Section~\ref{s:binomial}, we derive an error bound for the influence vector (with high probability) when the missing data has a binomial distribution.

In Section~\ref{s:refined}, we define networks of firms and goods, and we prove that the worst-case error for the influence vectors of any such network is $\Theta(\delta)$ when at most a $\delta$ share of the data for the intermediate input use of each firm is missing. We prove error bounds for chain-like networks of firms in Section~\ref{s:chains}. In Section~\ref{s:missdegrees}, we investigate what happens when we only have data on nodes and connections that are within distance $k$ of some source node. In Section~\ref{s:pagerank}, we discuss applications of our error bounds to \texttt{PageRank}. Finally in Section~\ref{s:future}, we discuss future directions for research on error bounds for importance measures of networks in the presence of missing data.

\section{Definitions and notation}\label{s:def}

The set $M_N(\R)$ consists of all $N \times N$ matrices with real entries. The \textit{general linear group} $\textnormal{GL}_N(\rr) \subseteq M_N(\R)$ consists of all invertible matrices in $M_N(\R)$. For any vector $v = (v_1, v_2, \dots, v_n) \in \mathbb{R}^n,$ the $p$-norm of $v$ is defined as 
\[
\norm{v}_p = \bigg(\sum_{i = 1}^n |v_i|^p \bigg)^{1/p}.
\] 
For all $p \ge 1$, the $p$-norm of a fixed vector $v$ is decreasing with respect to $p$. For any matrix $A: \mathbb{R}^m \rightarrow \mathbb{R}^n$ and $1 \le p \le \infty$, the $p$-norm of $A$ is defined as 
\[
\norm{A}_p = \max_{v \neq 0} \frac{\norm{A v}_p}{\norm{v}_p}.
\] 
We note that the $p$-norm is submultiplicative, i.e. $\norm{AB}_p \leq \norm{A}_p \norm{B}_p.$


For $\delta \in (0,1)$ and $f, g: (0,1) \rightarrow \mathbb{R}$, we say that $f(\delta) = O(g(\delta))$ if there exist constants $D, K > 0$ such that $f(\delta) \le K g(\delta)$ for all $\delta \in (0, D)$. Similarly, we say that $f(\delta) = \Omega(g(\delta))$ if there exist constants $D, K > 0$ such that $f(\delta) \ge K g(\delta)$ for all $\delta \in (0, D)$. Furthermore, we say that $f(\delta) = \Theta(g(\delta))$ if both $f(\delta) = O(g(\delta))$ and $f(\delta) = \Omega(g(\delta))$. 

For $f, g: \mathbb{R}^{+} \rightarrow \mathbb{R}^{+}$, we say that $f(n) = o(g(n))$ if \[\lim_{n \rightarrow \infty} \frac{f(n)}{g(n)} = 0.\] Similarly, we say that $f(n) = \omega(g(n))$ if \[\lim_{n \rightarrow \infty} \frac{f(n)}{g(n)} = \infty.\] 


In this paper, we assume that the data to compute the influence vector is incomplete, i.e., we only observe some of the data on the intermediate input use of firm $f_i$, and the remainder of the data is considered to be \textit{missing}. We prove bounds on the worst-case error for the influence vector in the presence of missing data. When we say that \textit{the missing data accounts for at most a $\delta$ share of the total intermediate input use of firm $f_i$}, we mean that firm $f_i$ has some amount $\text{II}_i$ of intermediate input use (including both observed and missing data), and at most $\delta \text{II}_i$ of the intermediate input use of firm $f_i$ is missing (not observed), while at least $(1-\delta) \text{II}_i$ of the intermediate input use of firm $f_i$ is observed.

\section{Estimating the influence vector with missing data}\label{s:error_bounds}

In the next theorem, we prove an upper bound on the error in the influence vector in terms of the share of missing data for each firm in the network. 

\begin{theorem}\label{maxdeltacor}
Let $W$ be the true matrix of input-output linkages with all of the data, and let $U$ be the observed matrix of the input-output linkages with data missing. Suppose that for each $i = 1, \dots, n$, the missing data accounts for at most a $\delta$ share of the total intermediate input use of firm $i$, where $0 \le \delta < 1$. If $v_U$ denotes the influence vector calculated with the observed data and $v_W$ denotes the true influence vector, then 
\[
\norm{v_U - v_W}_1 \le \frac{\delta}{\alpha}\frac{(1-\alpha)(2-\delta)}{1-\delta}.
\]
\end{theorem}

\begin{proof}
As in Theorem~\ref{thm:ml_pr}, suppose that $B = W-U$. We first derive an upper bound on $\sum_{j=1}^n |u_{i j} - w_{i j}|$. Observe that if the missing data accounts for a $d_i$ share of the goods in the total intermediate input use of firm $i$ for some $d_i \le \delta$, and a $c_{i j}$ share of the missing input data for firm $i$ is from the output of firm $j$, then we have 
\[
w_{i j} - u_{i j} = w_{i j} - \frac{w_{i j}-c_{i j}}{1-d_i} = \frac{c_{i j} - d_i w_{i j}}{1-d_i}.
\] 
Moreover, note that $c_{i j} \le w_{i j}$. Thus, \[\sum_{j=1}^n \max(w_{i j} - u_{i j},0) \le \sum_{j = 1}^n \frac{c_{i j}-d_i c_{i j}}{1-d_i} = \sum_{j = 1}^n c_{i j} = d_i.\] Furthermore, \[\sum_{j=1}^n \min(w_{i j} - u_{i j},0) \ge \sum_{j=1}^n \frac{- d_i w_{i j}}{1-d_i} = -\frac{d_i}{1-d_i}.\] Therefore
\[
\sum_{j=1}^n |u_{i j} - w_{i j}| \le d_i+\frac{d_i}{1-d_i} = \frac{2d_i-d_i^2}{1-d_i} \le \frac{2\delta-\delta^2}{1-\delta}.
\]
Thus, $\norm{B}_{\infty} \le \frac{2\delta-\delta^2}{1-\delta}$. Plugging this into Theorem~\ref{thm:ml_pr} implies 
\[
\norm{v_U - v_W}_1 \le \frac{(1-\alpha)(2\delta-\delta^2)}{\alpha (1-\delta)},
\]
as it was to be shown.
\end{proof}

\begin{corollary} Under the same conditions as Theorem~\textnormal{\ref{maxdeltacor}}, we have \[\norm{v_U - v_W}_p \le \frac{(1-\alpha)(2\delta-\delta^2)}{\alpha(1-\delta)}\] for all $p\ge 1$.
\end{corollary}

\begin{proof}
This follows from the fact that the $p$-norm is decreasing with respect to $p$.
\end{proof}

Next, we show that it is possible to construct an input-output network where the missing data accounts for at most a $\delta$ share of the goods in the total intermediate input use of each firm and the error in the influence vector is $\Omega(\delta)$, where the constants in the bound depend on $\alpha$. In the following result, $W$ represents an input-output matrix of $n$ firms with full data, and $U$ represents a corresponding input-output matrix of the same $n$ firms where at most a $\delta$ share of the goods in the total intermediate input use of each firm is missing. We let $v_W$ denote the influence vector of $W$ and $v_U$ denote the influence vector of $U$.

In order to prove the next result, we use the fact that the influence vector is the steady state of a random walk on the network of firms where the walker flips a coin which lands on heads with probability $1-\alpha$ and tails with probability $\alpha$. If the walker is at firm $i$ and the coin lands on heads, then they walk to firm $j$ with probability $w_{i j}$. If the coin lands on tails, then they walk to any firm in the network uniformly with probability $\frac{1}{n}$, including their current location. This steady state definition enables us to obtain a system of equations and solve for the influence vector.

\begin{theorem}\label{lower_bound_delta}
    For all $n > 2$ and all $p \ge 1$, it is possible to construct $W$ and $U$ such that $\norm{v_W - v_U}_p = \Omega(\delta)$, where the constants in the bound depend on $\alpha$.
\end{theorem}

\begin{proof}
    Let $W$ be the $n \times n$ matrix where $w_{1 i} = \frac{1}{n-1}$ for $i \ne 1$, $w_{2 i} = \frac{1}{n-1}$ for $i \ne 2$, $w_{j 1} = 1-\delta$ for $j > 2$, $w_{j 2} = \delta$ for $j > 2$, and all other entries are $0$. Let $U$ be the $n \times n$ matrix obtained from $W$ by setting $\delta$ to $0$. In other words, $u_{1 i} = \frac{1}{n-1}$ for $i \ne 1$, $u_{2 i} = \frac{1}{n-1}$ for $i \ne 2$, $u_{j 1} = 1$ for $j > 2$ and all other entries are $0$.

    Let $x$ be the first coefficient of the influence vector for $W$ and let $y$ be the second coefficient. By symmetry all other coefficients are equal and are denoted by $z$. 
    Using the fact that the influence vector is a steady state of the random walk described before this proof yields the following system of equations
    \begin{align}
    x &= (1-\alpha) \left((1-\delta) (n-2) z+\frac{y}{n-1}\right)+\frac{\alpha}{n}, \nonumber \\
   y &=(1-\alpha) \left(\delta (n-2) z+\frac{x}{n-1}\right)+\frac{\alpha}{n},, \nonumber \\
   z&=(1-\alpha) \left(\frac{x}{n-1}+\frac{y}{n-1}\right)+\frac{\alpha}{n}.
    \end{align}
    Let $a$ and $b$ be the first and second coefficients of the influence vector of $U$. Likewise, by symmetry all the other coefficients are the same and denoted by $c$. The same reasoning leads us to
    \begin{align}
    a &= (1-\alpha) \left((n-2) z+\frac{y}{n-1}\right)+\frac{\alpha}{n}, \nonumber \\
   b &=(1-\alpha) \frac{x}{n-1}+\frac{\alpha}{n}, \nonumber \\
   c&=(1-\alpha) \left(\frac{x}{n-1}+\frac{y}{n-1}\right)+\frac{\alpha}{n} .
    \end{align}
The solution to the first system is given by
\begin{align}
    x=x_n(\delta) &= \frac{(n-1) (n-1-\alpha (1-\delta) (n-2)-\delta (n-2))}{n (2n-3-\alpha (n-2))}, \nonumber \\
    y=y_n(\delta) &= \frac{(n-1) ((1-\alpha) \delta (n-2)+1)}{n (2n-3-\alpha (n-2))}, \nonumber \\
    z=z_n(\delta) &= \frac{n-\alpha}{n (2n-3-\alpha (n-2))}.
\end{align}
Note that $a = x_n(0)$, $b = y_n(0)$, and $c = z_n(0)$. Thus, we see that
\begin{align}
    x-a &= -\frac{(1-\alpha) \delta (n-2) (n-1)}{n (2n-3-\alpha (n-2))}, \nonumber \\
    y-b &= \frac{(1-\alpha) \delta (n-2) (n-1)}{n (2n-3-\alpha (n-2))}.
\end{align}
As $n \to \infty$, we can write
\begin{align}
    x_n(\delta) &= \frac{1-\alpha}{2-\alpha}+\frac{-\alpha^2+3 \alpha-1}{(2-\alpha)^2
   n}+O\bigg(\frac{1}{n^2}\bigg), \nonumber \\
   y_n(\delta) &= \frac{1}{(2-\alpha) n}+\frac{1-\alpha}{(2-\alpha)^2
   n^2}+O\bigg(\frac{1}{n^3}\bigg).
\end{align}
This in turn allows us to see that
\begin{align}
    \lim_{n \to \infty} (x-a) = -\frac{1-\alpha}{2-\alpha}\delta \quad \textnormal{and} \quad \lim_{n \to \infty} (y-b) = \frac{1-\alpha}{2-\alpha}\delta.
\end{align} 
Thus, $\norm{v_W - v_U}_p = \Omega(\delta)$, completing the proof.
\end{proof}

Note that as $n$ goes to $\infty$, the coefficient of $\delta$ in the last proof is on the order of $1-\alpha$. However, the coefficient of $\delta$ in our upper bound was on the order of $\frac{1-\alpha}{\alpha}$. The gap between our upper and lower bounds is on the order of $\alpha$. It is an open problem to decrease this gap. 

Theorem~\ref{lower_bound_delta} tells us that the error of our importance measure may be at least on the order of the share of missing data. Thus if we want to accurately determine the most important firms in the network, we should not have too much missing data.

Combining the upper bound in Theorem~\ref{maxdeltacor} with the lower bound in Theorem~\ref{lower_bound_delta}, we have the following result which is sharp up to a multiplicative factor that is on the order of $\alpha$.

\begin{theorem}\label{thmsharpdelta}
Let $W$ be the true matrix of input-output linkages with all of the data, and let $U$ be the observed matrix of the input-output linkages with data missing. Suppose that for each $i = 1, \dots, n$, the missing data accounts for at most a $\delta$ share of the goods in the total intermediate input use of firm $i$, where $0 \le \delta < 1$. If $v_U$ denotes the influence vector calculated with the missing data and $v_W$ denotes the true influence vector, then the maximum possible value of $\norm{v_U - v_W}_p$ is $\Theta(\delta)$ for all $p \ge 1$, where the constants in the bound depend on $\alpha$.
\end{theorem}

To put this result in perspective, note that the maximum possible value of $\norm{v_W}_p$ is $1$ since the $p$-norm is decreasing with respect to $p$ and $\norm{v_W}_1 = 1$. Moreover, the minimum possible value of $\norm{v_W}_p$ is $\frac{1}{n^{1-1/p}}$ by the power mean inequality \cite{power_mean}. 

For more perspective, we derive extremal bounds on the coefficients of the influence vector in the next section. In particular, we see that the maximum possible value of any coefficient is approximately equal to the value of $x$ in the proof of Theorem~\ref{lower_bound_delta}.

\section{Extremal bounds on coefficients of the influence vector}\label{s:extremal}

In this section, we prove some extremal bounds on the values of the coefficients of the influence vector. Our proofs in this section also use the fact which we mentioned in the last section that the influence vector is the steady state of a random walk on the network of firms. 

\begin{theorem}\label{extremalupper}
In a network of $n$ firms, the maximum possible value of any coefficient in the influence vector is \[\frac{1-\frac{n-1}{n}\alpha}{2-\alpha}.\]    
\end{theorem}

\begin{proof}
We start by proving an upper bound on the value of any coefficient of the influence vector, and then we construct a network of $n$ firms that attains the upper bound. For the upper bound, we use two facts. First, as we mentioned before, the influence vector is the steady state of a random walk on the network of firms. Second, the sum of the coefficients of the influence vector is $1$.

Let $p$ be the value of some coefficient of the influence vector. Without loss of generality, we will assume that it is the first coefficient. Since the sum of the coefficients of the influence vector must be equal to $1$, we know that the coefficients besides the first coefficient must sum up to $1-p$. Suppose that the walker is in the steady state. If the walker's coin lands heads, then the probability that their next location will not be the first firm is at least $p$. This is because they have probability $p$ of currently being on the first firm, and if they are currently on the first firm, then they have probability $1$ of being at a different firm on the next turn since $w_{i i} = 0$ for all $i$. If the walker's coin lands tails, then the probability that their next location will not be the first firm is equal to $\frac{n-1}{n}$, since every firm has probability $\frac{1}{n}$ of being the next location if the coin lands tails. 

Thus, overall the probability that the walker's next location will not be the first firm is at least $(1-\alpha)p + \alpha \frac{n-1}{n}$. Since the walker is in the steady state, the probability that their next location will not be the first firm is equal to $1-p$, so we have \[(1-\alpha)p + \alpha \frac{n-1}{n} \le 1-p.\] Solving the inequality for $p$, we obtain \[p \le \frac{1-\frac{n-1}{n}\alpha}{2-\alpha}.\] This completes the proof of the upper bound.

Now, we construct a network of $n$ firms which attains the upper bound. Consider the network with firms $1, 2, \dots, n$ where $w_{1 i} = \frac{1}{n-1}$ and $w_{i 1} = 1$ for all $i \ge 2$, and all other entries of $W$ are $0$. Let $x$ be the value of the first coefficient of the influence vector of this network, and let $y$ be the value of the other coefficients. By symmetry, all coefficients besides the first have the same value. Since the influence vector is the steady state, we have \[x = (n-1)y(1-\alpha)+\frac{\alpha}{n} \quad \textnormal{and} \quad y = (1-\alpha)\frac{x}{n-1}+\frac{\alpha}{n}.\] When we solve this system of equations, we obtain \[x = \frac{1-\frac{n-1}{n}\alpha}{2-\alpha}.\] Thus we have obtained the maximum possible value of any coefficient in the influence vector.
\end{proof}

Note that as $n$ approaches $\infty$, the maximum possible value of any coefficient of the influence vector approaches $\frac{1-\alpha}{2-\alpha}$. In addition to the network constructed in the proof of Theorem~\ref{extremalupper}, we describe another network of firms which attains the maximum possible coefficient value. 

\begin{remark}\label{altmaxcoeff}
  Consider the network with firms $1, 2, \dots, n$ for $n \ge 3$ where $w_{1 2} = 1$, $w_{i 1} = 1$ for all $i \ge 2$, and all other entries of $W$ are $0$. By symmetry, all coefficients of the influence vector are the same besides the first two coefficients. Let $x$ be the value of the first coefficient, let $y$ be the value of the second coefficient, and let $z$ be the value of the other coefficients. For this network, we obtain the system of equations:

\begin{align}
    x &= (1-\alpha)((n-2)z + y)+\frac{\alpha}{n}, \nonumber \\
    y &= (1-\alpha)x + \frac{\alpha}{n}, \nonumber \\
    z &= \frac{\alpha}{n}.
\end{align}

When we solve the system of equations, we obtain \[x = \frac{1-\frac{n-1}{n}\alpha}{2-\alpha} \quad \textnormal{and} \quad y = \frac{(1-\alpha)(1-\frac{n-1}{n}\alpha)}{2-\alpha}+\frac{\alpha}{n}.\]  
\end{remark}

Now we turn to the minimum possible value of any coefficient of the influence vector, which is attained by the network described in Remark~\ref{altmaxcoeff}. Using the fact that the influence vector is a steady state of the walk described at the beginning of this section, it is relatively easy to prove that $\frac{\alpha}{n}$ is the minimum possible value of any coefficient. 

\begin{theorem}
The minimum possible value of any coefficient in the influence vector is $\frac{\alpha}{n}$ for $n \ge 3$.
\end{theorem}

\begin{proof}
We start by proving a lower bound on the value of any coefficient of the influence vector, and then we construct a network of $n$ firms that attains the lower bound. For the lower bound, note that for any firm $i$ in the network, the probability that the walker will be on firm $i$ in the next round is $\frac{1}{n}$ if their coin lands tails. Since the coin lands tails with probability $\alpha$, the walker will be at firm $i$ in the next round with probability at least $\frac{\alpha}{n}$. This proves the lower bound. 

When $n = 2$, note that we must have $w_{1 1} = w_{2 2} = 0$ and $w_{1 2} = w_{2 1} = 1$, so by symmetry both coefficients of the influence vector have value $\frac{1}{2}$. When $n = 3$, we use the network defined in Remark~\ref{altmaxcoeff}, for which the value of the $i^{th}$ coefficient of the influence vector is $\frac{\alpha}{n}$ for all $i \ge 3$.
\end{proof}

\section{A negative result: share of firms versus share of intermediate input use}\label{s:negative}
Instead of considering scenarios where the missing data accounts for at most a $\delta$ share of the intermediate input use of each firm, in this section we examine what happens when the missing data accounts for at most a $\delta$ share of the firms in the network. Unfortunately, in this scenario we find that the error in the  influence vector does not go to $0$ as $\delta$ goes to $0$. Note that this is quite different from what we see when the missing data accounts for at most a $\delta$ share of the intermediate input use of each firm, since in that scenario the error in the influence vector scales at most linearly with $\delta$.

In the following theorem, we show that there exist networks of firms where one firm has influence vector coefficient with nearly maximum possible value and all other coefficients are close to $0$, but with missing data a different firm has a coefficient with nearly maximum possible value and all other coefficients are close to $0$.

\begin{theorem}\label{deltafirmshare}
For every $\delta > 0$, there exists $N$ such that for all $n \ge N$, there is a network $G$ of $n$ firms with connectivity matrix parameterized by a variable $\epsilon > 0$ with the following properties:
\begin{enumerate}
    \item The value of the influence vector coefficient of the first firm approaches $\frac{1-\alpha}{2-\alpha}$ and all other coefficients approach $0$ as $\epsilon$ goes to $0$ and $n$ goes to $\infty$.
    \item There exists a scenario where we are missing data for at most a $\delta$ share of the firms in $G$, and in the observed network $G'$ the influence vector coefficient of the second firm approaches $\frac{1-\alpha}{2-\alpha}$ and all other coefficients approach $0$ as $\epsilon$ goes to $0$ and $n$ goes to $\infty$.
\end{enumerate}
\end{theorem}

\begin{proof}
First, we choose $N \ge \frac{1}{\delta}$, and suppose that $n \ge N$. Consider the network $G$ of $n$ firms where $w_{1 i} = \frac{1}{n-1}$ for $i \ne 1$, $w_{2 i} = \frac{1}{n-1}$ for $i \ne 2$, $w_{i 1} = 1-\epsilon$ for all $i \ge 3$, $w_{i 2} = \epsilon$ for all $i \ge 3$, and all other entries of $W$ are $0$. By Theorem~\ref{maxdeltacor}, as $\epsilon$ approaches $0$ the influence vector of network $G$ converges to the influence vector of the network defined in Theorem~\ref{lower_bound_delta} with connectivity matrix $U$. Thus, the value of the influence vector coefficient of the first firm in $G$ goes to $\frac{1-\alpha}{2-\alpha}$ as $\epsilon$ goes to $0$ and $n$ goes to $\infty$, while all other coefficients go to $0$. This completes the proof of the first part of the theorem statement.

For the second part, suppose that for each $i \ge 3$, we are missing data on the contribution of firm $1$ of $G$ to the production of firm $i$. Since $n \ge N \ge \frac{1}{\delta}$, we have $\frac{1}{n} \le \delta$, so we are missing data for at most a $\delta$ share of the firms in the network. With the missing data, the observed network $H$ has connectivity matrix $Q$ with $q_{1 i} = \frac{1}{n-1}$ for $i \ne 1$, $q_{2 i} = \frac{1}{n-1}$ for $i \ne 2$, $q_{i 2} = 1$ for all $i \ge 3$, and all other entries of $Q$ are $0$. Note that $Q$ is the same as the connectivity matrix $U$ in Theorem~\ref{lower_bound_delta}, except the first two firms have been transposed. Thus, the value of the influence vector coefficient of the second firm in $H$ goes to $\frac{1-\alpha}{2-\alpha}$ as $\epsilon$ goes to $0$ and $n$ goes to $\infty$, while all other coefficients go to $0$. This completes the proof of the second part of the theorem statement.
\end{proof}

In Theorem~\ref{deltafirmshare}, note that the $p$-error in the influence vector is at least on the order of $\frac{1-\alpha}{2-\alpha}$ as $\epsilon$ goes to $0$ and $n$ goes to $\infty$, since the first two coefficients of the influence vector both have an error which goes to $\frac{1-\alpha}{2-\alpha}$ as $\epsilon$ goes to $0$ and $n$ goes to $\infty$. It is an open problem to determine the maximum possible $p$-error in the influence vector when we are missing data for at most a $\delta$ share of the firms.

\section{Binomial distribution of missing data}\label{s:binomial}

In this section, we bound the error in the influence vector when the missing data has a binomial distribution. In order to state and prove the error bound, we introduce some notation.

\begin{definition}
Recall that $\text{II}_i$ is the intermediate input use of firm $i$. We will assume that $\text{II}_i$ is measured in USD. 
\begin{itemize}
    \item Let $y_{ij}= w_{ij} \text{II}_i$. This corresponds to the amount of intermediate input use for firm $i$ coming from firm $j$. 
    \item We assume that each dollar $j=1,2,\cdots, \text{II}_i$ is observed with probability $1-\zeta$ and is otherwise missing with probability $\zeta$.
    \item Define $M$ so that that $y_{ij} \ge M$ for all $i,j$, e.g., we can take $M$ to be the minimum $y_{ij}$ over all $i, j$.
    \item For each $i,j$ we define the random variable $X_{ij}$ to be the amount of observed USD from the $y_{ij}$ pool of USD.
    \end{itemize}
\end{definition}

By definition, note that $X_{ij}\sim \operatorname{Bin}(r,p)$ where $r=y_{ij}$ and $p=1-\zeta$. Thus, $\mathbb{E}[X_{ij}] = (1-\zeta)y_{ij}$ for each $i, j$ and \[\mathbb{P}(X_{ij}=k) = \binom{y_{ij}}{k} (1-\zeta)^k \zeta^{y_{ij}-k}.\]

\begin{theorem}\label{thm:binomial}
Let $W$ be the true matrix of input-output linkages with all of the data, and let $U$ be the observed matrix of the input-output linkages with data missing. Suppose that for each $i, j$, the observed data $X_{i j}$ satisfies $X_{ij}\sim \operatorname{Bin}(r,p)$ with $r=y_{ij}$ and $p=1-\zeta$ for some $0 < \zeta < 1$. If $v_U$ denotes the influence vector calculated with the observed data and $v_W$ denotes the true influence vector, then for all $q \ge 1$ and $0 < \varepsilon < 1$ we have \[\mathbb{P}\left(\norm{v_U - v_W}_q \le \frac{2(1-\alpha)\varepsilon}{\alpha(1-\varepsilon)}\right) \ge 1-2n^2 e^{-\frac{\varepsilon^2}{3}(1-\zeta)M}.\]
\end{theorem}
\begin{proof}
    For each $i$, by Chernoff's bound \cite{chernoff} we have
\[\mathbb{P}\left(X_{ij} \ge (1+\varepsilon)(1-\zeta)y_{ij}\right) \le e^{-\frac{\varepsilon^2}{3}(1-\zeta)y_{ij}} \quad \textnormal{and} \quad \mathbb{P}\left(X_{i} \le (1-\varepsilon)(1-\zeta)y_{ij}\right) \le e^{-\frac{\varepsilon^2}{3}(1-\zeta)y_{ij}}.
\] Thus, we have
\[\mathbb{P}\left(|X_{ij} - (1-\zeta)y_{ij}| \ge \varepsilon(1-\zeta)y_{ij}\right) \le 2 e^{-\frac{\varepsilon^2}{3}(1-\zeta)y_{ij}} \le 2 e^{-\frac{\varepsilon^2}{3}(1-\zeta)M}.
\] By the union bound, we have 
\begin{align}
\mathbb{P}\left(|X_{ij} - (1-\zeta)y_{ij}| \ge \varepsilon(1-\zeta)y_{ij} \text{ for some }i, j\right) &= \mathbb{P}\bigg(\bigvee_{1 \le i, j \le n} \{|X_{ij} - (1-\zeta)y_{ij}| \ge \varepsilon(1-\zeta)y_{ij}\}\bigg) \nonumber \\
&\le \sum_{1 \le i, j \le n} \mathbb{P}\left(|X_{ij} - (1-\zeta)y_{ij}| \ge \varepsilon(1-\zeta)y_{ij}\right) \nonumber \\ &\le 2e^{-\frac{\varepsilon^2}{3}(1-\zeta)M} \sum_{1 \le i, j \le n} 1 =  2n^2 e^{-\frac{\varepsilon^2}{3}(1-\zeta)M},
\end{align} 
where $\bigvee_{k \in \mathcal{K}} E_k$ denotes the disjunction of events $E_k$ with $k$ indexed over an arbitrary set $\mathcal{K}$. Thus, using the double-sided inequalities 
\[
(1-\varepsilon)(1-\zeta)y_{ij} \le X_{ij} \le (1+\varepsilon)(1-\zeta)y_{ij}
\] 
as well as 
\[
\sum_{1 \le k \le n} (1-\varepsilon)(1-\zeta)y_{ik} \le \sum_{1 \le k \le n} X_{ik} \le \sum_{1 \le k \le n} (1+\varepsilon)(1-\zeta)y_{ik},\] we have
\[u_{ij} = \frac{X_{ij}}{\sum_{k = 1}^n X_{ik}} \le \frac{(1+\varepsilon)(1-\zeta)y_{ij}}{\sum_{k = 1}^n (1-\varepsilon)(1-\zeta)y_{ik}} = \frac{1+\varepsilon}{1-\varepsilon} w_{ij},\] and \[u_{ij} = \frac{X_{ij}}{\sum_{k = 1}^n X_{ik}} \ge \frac{(1-\varepsilon)(1-\zeta)y_{ij}}{\sum_{k = 1}^n (1+\varepsilon)(1-\zeta)y_{ik}} = \frac{1-\varepsilon}{1+\varepsilon} w_{ij},\] for all $i, j$ with probability at least 
\begin{align} \label{eq:auxProb}
1-2n^2 e^{-\frac{\varepsilon^2}{3}(1-\zeta)M}.
\end{align}
Therefore, 
\[u_{ij}-w_{ij} \le \left( \frac{1+\varepsilon}{1-\varepsilon} - 1\right) w_{ij} = \frac{2\varepsilon}{1-\varepsilon} w_{ij} \text{ and } w_{ij} - u_{ij} \le \left( 1-\frac{1-\varepsilon}{1+\varepsilon}\right) w_{ij} = \frac{2\varepsilon}{1+\varepsilon} w_{ij}.\]
Consequently, \[|u_{ij}-w_{ij}| \le \max\bigg(\frac{2\varepsilon}{1-\varepsilon} w_{ij},\frac{2\varepsilon}{1+\varepsilon} w_{ij}\bigg) = \frac{2\varepsilon}{1-\varepsilon} w_{ij}. \]
Thus, we have
\[
\mathbb{P}\bigg(\sum_{j = 1}^n |u_{i j} - w_{i j}| \le \frac{2\varepsilon}{1-\varepsilon} \sum_{j = 1}^n w_{ij} = \frac{2\varepsilon}{1-\varepsilon} \text{ for all $i$}\bigg) \ge 1-2n^2 e^{-\frac{\varepsilon^2}{3}(1-\zeta)M},
\] 
by recalling \eqref{eq:auxProb}. Let $B = W-U$. Then we have \[\mathbb{P}\left(\norm{B}_{\infty} \le \frac{2\varepsilon}{1-\varepsilon} \right) \ge 1-2n^2 e^{-\frac{\varepsilon^2}{3}(1-\zeta)M}.\] Thus, by Theorem~\ref{thm:ml_pr} we have
\[\mathbb{P}\left(\norm{v_U - v_W}_1 \le \frac{2(1-\alpha)\varepsilon }{\alpha(1-\varepsilon)}\right) \ge 1-2n^2 e^{-\frac{\varepsilon^2}{3}(1-\zeta)M}.\] Therefore, \[\mathbb{P}\left(\norm{v_U - v_W}_q \le \frac{2(1-\alpha)\varepsilon}{\alpha(1-\varepsilon)}\right) \ge 1-2n^2 e^{-\frac{\varepsilon^2}{3}(1-\zeta)M},\]
as it was to be shown.
\end{proof} 

\begin{corollary} \label{cor:differencenormQ}
    Suppose that $M = \omega(\ln n)$. Under the same assumptions as Theorem~\textnormal{\ref{thm:binomial}}, we have \[\mathbb{P}\left(\norm{v_U - v_W}_q = o(1) \right) \ge 1-o(1).\]
\end{corollary}

\begin{proof}
    Take $\varepsilon = \sqrt{\frac{7\ln n}{M(1-\zeta)}}$ in Theorem~\ref{thm:binomial}.
\end{proof}
We remark that the integer $7$ appearing in the choice of $\varepsilon$ in the proof of Corollary \ref{cor:differencenormQ} could be replaced by any real $\mathfrak{c}$ with $\mathfrak{c} > 6$.

\section{Influence vector for networks of firms and goods}\label{s:refined}

We define a weighted directed network $G = (V, E)$ whose vertex set $V$ consists of ordered pairs $(f, g)$ where $f$ is a firm and $g$ is a good produced by firm $f$. Suppose that $|V| = N$ and we index the vertices $v_1, v_2, \dots, v_N$ where $v_i = (f_i, g_i)$. Note that for a given firm $f$, multiple vertices $v_i = (f_i, g_i)$ can have $f_i = f$ since firm $f$ may produce multiple goods. Moreover for a given good $g$, multiple vertices $v_i = (f_i, g_i)$ can have $g_i = g$ since multiple firms may produce good $g$.

In the network $G$, there is a directed edge $(v_i, v_j)$ from $v_i = (f_i, g_i)$ to $v_j = (f_j, g_j)$ if the good $g_j$ from firm $f_j$ is used in the production of the good $g_i$ from firm $f_i$. We define some terminology in order to explain how the edges are weighted. We say that good $g_j$ from firm $f_j$ has \textit{criticality} $c_{i, j}$ to the production of good $g_i$ from firm $f_i$, where the criticality is proportional to the loss of production of good $g_i$ from firm $f_i$ if good $g_j$ from firm $f_j$ was unavailable.\footnote{A different way to define the criticality $c_{i, j}$ is the cost that firm $f_i$ pays firm $f_j$ to acquire the good $g_j$ for the production of good $g_i$. The choice of definition would depend on how one wants to measure importance.} The criticality $c_{i, j}$ is $0$ if good $j$ from firm $j$ is not used in the production of good $i$ from firm $i$. The edge $(v_i, v_j)$ has weight 
\begin{align} \label{eq:weightij}
w_{i j} = \frac{c_{i, j}}{\sum_{k = 1}^{N} c_{i, k}}.
\end{align}

We want to know the most important firms in the network, so we will use the same general idea as the influence vector of Acemoglu \textit{et al.} \cite{acemoglu} and the \texttt{PageRank} vector of Google \cite{pagerank}. We consider a random walker on the network. We define a random process that defines how the walker moves between vertices, so that each step with high probability takes the walker from their current vertex $v_i = (f_i, g_i)$ to some vertex $v_j = (f_j, g_j)$ for which good $g_j$ from firm $f_j$ is used in the production of the good $g_i$ from firm $f_i$. The probability of moving from $v_i$ to $v_j$ will be proportional to the criticality $c_{i, j}$. 

We determine the probability that the walker is at each vertex as the number of steps approaches infinity. From these probabilities, we can determine for each firm $f$ the probability that the walker is at any of the vertices with first coordinate $f$ as the number of steps goes to infinity, just by summing the probabilities of the vertices with first coordinate $f$. The probabilities for each of the firms are their importance scores, with the intuition for the importance being that the random walker will have a higher probability of being at some vertex $v_j = (f_j, g_j)$ if good $g_j$ from firm $f_j$ is critical to the production of goods in the network.

Now we specify the details of the random walk. Initially the walker chooses a vertex uniformly at random from the network, i.e. each vertex has probability $\frac{1}{N}$ of being chosen. Note that the probability distribution for the initial position is not important, as it will have no effect on the final scores. For each move, the walker chooses a step type by flipping a biased coin which lands heads with probability $1-\alpha$ and tails with probability $\alpha$. If the coin lands heads and the current vertex is $v_i = (f_i, g_i)$, then the walker goes with probability $w_{i j}$ given by \eqref{eq:weightij} to vertex $v_j = (f_j, g_j)$ for which $(v_i, v_j) \in E$. If the coin lands tails, then the walker goes to any vertex in the network uniformly at random, i.e. every vertex in the network has probability $\frac{1}{N}$ of being chosen. 

Given the matrix $W$ with entries $w_{i j}$ for $1 \leq i, j \leq n$, we find the vector $u$ for which $u_i$ is the limit of the probability that the random walker is at vertex $v_i$ as the number of rounds approaches infinity. This is equivalent to finding the limit vector of the random process where we start with the vector $z = \frac{1}{N}\vec{\mathbf{1}}$ of length $N$ with all coordinates equal to $\frac{1}{N}$, and then we repeatedly iterate the function $f(z) = \frac{\alpha}{n} \vec{\mathbf{1}} + (1-\alpha) W' z$. This limit vector can be determined by finding the fixed point $f(u) = u$, which gives us the influence vector  \[u = \frac{\alpha}{n}[I-(1-\alpha)W']^{-1}\vec{\mathbf{1}}.\] Note that this formula is identical to the influence vector defined in \cite{acemoglu} and \cite{pagerank}, with the difference being that we are calculating the influence vector on a different network. Given the influence vector $u$ for the network $G$ where $u_i$ is the probability that the random walker is at vertex $v_i = (f_i, g_i)$ as the number of steps approaches infinity, the probability that the random walker is at some vertex with first coordinate $f$ is 
\[
q_f = \sum_{\substack{1 \le i \le n \\ (f_i, g_i): f_i = f}} u_i
\]
as the number of steps goes to infinity. We define the importance score of firm $f$ to be $q_f$, and we call $q = (q_f)_{f \in \mathcal{F}}$ the \textit{importance score vector} of the set of firms $\mathcal{F}$.

Note that we can use Theorem~\ref{maxdeltacor} and Theorem~\ref{lower_bound_delta} to obtain a sharp bound on the maximum possible error when we calculate the importance score vector $q$ with missing data.

\begin{theorem}
Let $W$ be the true matrix of input-output linkages with all of the data for the network of $N$ firm-good ordered pairs, and let $U$ be the observed matrix of the input-output linkages with data missing for the network of $N$ firm-good ordered pairs. Suppose that for each $i = 1, \dots, N$, the missing data accounts for at most a $\delta$ share of the total intermediate input use of the good $g_i$ from firm $f_i$, where $0 \le \delta < 1$. If $q_U$ denotes the firm importance score vector calculated with the missing data and $q_W$ denotes the true firm importance score vector, then the maximum possible value of $\norm{q_U - q_W}_p$ is $\Theta(\delta)$ for all $p \ge 1$, where the constants in the bound depend on $\alpha$.
\end{theorem}

\begin{proof}
The lower bound of $\Omega(\delta)$ follows from Theorem~\ref{lower_bound_delta}, since we can consider a network of firm-good ordered pairs in which every firm produces exactly one good and use the exact same construction as in the proof of Theorem~\ref{lower_bound_delta}. For the upper bound of $O(\delta)$, first note that $\norm{u_U - u_W}_1 = O(\delta)$ by Theorem~\ref{maxdeltacor}, where $u_U$ is the influence vector of the network of firm-good ordered pairs calculated with missing data and $u_W$ is the true influence vector of the network of firm-good ordered pairs. In other words, \[\sum_{i = 1}^{N} |u_{U,i}-u_{W,i}| = O(\delta).\] Since $q_f = \sum_{(f_i, g_i): f_i = f} u_i$ for each firm $f$, we have 
\[
\sum_{f \in \mathcal{F}} |q_{U,f}-q_{W,f}| = \sum_{f \in \mathcal{F}} \bigg|\sum_{(f_i, g_i): f_i = f} u_{U,i}-u_{W,i} \bigg| \le \sum_{f \in \mathcal{F}} \sum_{(f_i, g_i): f_i = f} |u_{U,i}-u_{W,i}| = \sum_{i = 1}^{N} |u_{U,i}-u_{W,i}| = O(\delta).
\]
Thus $\norm{q_U - q_W}_1 = O(\delta)$. Since the $p$-norm is decreasing with respect to $p$, we have $\norm{q_U - q_W}_p = O(\delta)$ for all $p \ge 1$.
\end{proof}

\section{Chains of firms and goods}\label{s:chains}

In this section, we consider a special case when the network of firms and goods can be partitioned into weakly coupled parts, in a sense we define rigorously in Definition \ref{def:weakcoupling}. Under certain assumptions on the connectivity of the firm network, we derive a locality statement for estimating influence vectors with missing data. To do this, we first define a spatial structure on the vector space and algebra of operators associated with the network. 


As before, let $G$ denote a directed graph with edges $E$ between nodes $V = \set{f_1,\ldots, f_N}$. The nodes $V$ denote firms. For each directed edge $(f_i, f_j) \in E$, associate a weight $w_{ij} \geq 0$ with the assumption that for all $i$, $\sum_j w_{ij} =1$. 


Let $\mc{H} = \rr^{N}$ and define an orthonormal basis such that $e_i$ is an orthonormal basis vector associated to $f_i$.  For $X \subset V$, associate the subspace $\mc{H}_X$ as the span of all $e_i$ such that $f_i \in X$. We note that if $A : \mc{H}_X \to \mc{H}_X$ is a matrix, then $A$ extends to a matrix on $\mc{H}$ by $\tilde{A} = A \oplus 0 $, acting as $0$ on the orthogonal complement of $\mc{H}_X$. In the following, we freely identify $A$ with its extension to $\mc{H}$. If $Y \subset V$ and $Y \cap X = \varnothing$, then $\mc{H}_X$ and $\mc{H}_Y$ are orthogonal subspaces of $\mc{H}$. Given an operator $B: \mc{H}_Y \to \mc{H}_Y$, we denote by $A\oplus B : \mc{H}_X \oplus \mc{H}_Y \to \mc{H}_X \oplus \mc{H}_Y$ the direct sum of $A$ and $B$. In the following, we also freely write $A \oplus B$ in block matrix notation, 

\begin{align}
    A \oplus B = \begin{bmatrix}
        A & 0 \\ 0 & B
    \end{bmatrix},
\end{align}
where the $0$ entries are appropriately sized arrays of $0$s. Abusing notation, we will also use block matrix notation to refer to the extension of $A\oplus B$ to $\mc{H}$. 


Let $P_{ij}$ denote the rank one matrix $u \mapsto \ip{e_j, u} e_i $. In this notation, the input-output matrix is $W = \sum_{i,j} w_{ij} P_{ij} $. For matrices, we reserve $\norm{\cdot}$ to denote the spectral norm, which we will refer to as the operator norm. Otherwise, for vectors, $\norm{\cdot}$ denotes the Euclidean norm. 


\noindent \textbf{Bi-partition of graph}: Consider a nontrivial partition $V = V_1 \cup V_2$ with $V_1 \cap V_2 =  \varnothing$. We consider the matrices
\begin{align}\label{eq:two-case}
W_1 = [ w^{(1)}_{ij}] , ~~ W_2 = [w^{(2)}_{ij}], ~~ \textnormal{and} ~~ W_{\operatorname{int}} = [ W^{(\operatorname{int})}_{ij}],
\end{align}
where the matrix entries are defined as follows. For $n=1,2$, $w^{(n)}_{ij} = 0$ if $\set{f_i, f_j} \not \subset V_n $ and $w^{(n)}_{ij} = w_{ij}$ otherwise, i.e. $\set{ f_i, f_j }\subset V_n$.  Let $W_{\operatorname{int}}$ be defined by $W_{\operatorname{int}} = W - W_1 - W_2$. 


We note a couple of immediate properties. First, we prove that $(I - (1-\alpha)W_n')^{-1}$ exists for any $\alpha \in (0,1)$ and $n=1,2$.

\begin{lemma}\label{lem:inverse}
Suppose $U = [u_{ij}] \in M_N(\rr) $, $u_{ij}\geq 0$ for all $i,j$, and $\sum _j u_{ij} < 1$. Then $I- U'$ is invertible.
\end{lemma}

\begin{proof}
Since the maximum row sum of $U$ is less than $1$, $\norm{U'}_1 < 1$. But this implies that $\sum_{k =0}^\infty (U')^k$ is a convergent series in the metric topology from $\norm{\cdot}_1$. Since all norms on finite dimensional vector spaces are equivalent, this implies $S = \sum _{k=0}^\infty (U')^k$ is an $N\times N$ matrix. But this is the Neumann series for the inverse of $I - U'$, with
\begin{align}
S (I - U') =  (I - U')S = I ,
\end{align}
and hence $I-U'$ can indeed be inverted.
\end{proof}

Next, we show that the $W_n$ are operators on orthogonal subspace of $\mc{H}$. 

\begin{lemma}\label{lem:directsum}
Consider $\mc{H} = \mc{H}_1 \oplus \mc{H}_2$, where $\mc{H}_n = \operatorname{span}( e_k : f_k \in V_n).$ Then $\operatorname{ran}(W_n) \subset \mc{H}_n$ for $n=1,2$. Additionally, $\mc{H}_2 \subset \operatorname{ker}(W_1)$ and $\mc{H}_1 \subset \operatorname{ker}(W_2)$. Lastly, $W_1 W_2 = W_2 W_1 =0$. 
\end{lemma} 

\begin{proof}
The proof of the kernel and image inclusion is the same for $W_1$ and $W_2$, so consider only $W_1$. By definition, 
\begin{align}
W_1 = \sumtwo_{ i,j: f_i, f_j \in V_1} w_{ij} P_{ij} 
\end{align}
and each $P_{ij}$ has kernel $\operatorname{span}(e_j)^\perp$ and range $\operatorname{span}(e_i)$. Hence $\mc{H}_2 \subset \operatorname{ker}(W_1)$ and $\operatorname{ran}(W_1) \subset \mc{H}_1$. 


The equation $W_1 W_2 = W_2 W_1 =0$ follows from the fact that if $f_i \in V_1$ and $f_a \in V_2$, then $\ip{ e_i, e_a} = 0$; since then, writing 
\[
W_1 = \sumtwo_{i,j} w_{ij} P_{ij} \quad \textnormal{and} \quad W_2 = \sumtwo_{a,b} w_{ab} P_{ab},
\]
we get
\begin{align}
W_1 W_2 = \sumtwo_{\substack{ i,j : f_i, f_j \in V_1 \\ a,b : f_a,f_b\in V_2}} w_{ij} w_{ab} P_{ij}P_{ab} = 0 = W_2 W_1, 
\end{align} 
which is what we wanted to show.
\end{proof}

By relabelling the basis if necessary, we can represent the direct sum in Lemma \ref{lem:directsum} as a block diagonal matrix. We write
\begin{align}
W_1 = \begin{bmatrix} \hat{W}_1 & 0 \\ 0 & 0\end{bmatrix}, ~~~ W_2 = \begin{bmatrix} 0 & 0 \\ 0 & \hat{W}_2 \end{bmatrix}
\end{align}
where the $0$ entries are appropriately sized matrices of $0$s, not necessarily square, and we identify $W_i$ with $\hat{W}_i$.

\begin{definition}
Consider the input-output matrix $W = (W_1 \oplus W_2) + W_{\operatorname{int}}$ and Leontief inverse $L_\perp = (I - (1-\alpha)W_1' - (1-\alpha)W_2')^{-1} $, which exists by Lemma \ref{lem:inverse}. Define the \textit{interaction matrix} 
\begin{align}
S = (I - (1-\alpha)L_\perp W_{\operatorname{int}}')^{-1}.
\end{align}
\end{definition}
The matrix $S$ is well-defined since 
\begin{align}
(I - (1-\alpha) L_\perp W_{\operatorname{int}}') = L_\perp (I - (1-\alpha)W')  \in \textnormal{GL}_N(\rr).
\end{align}

\begin{lemma}\label{lem:basecase}
Denote $|V_n| = m_n$. Denote by $v_1, v_2$ the influence vectors $v_n = \frac{\alpha}{m_n} (I_n - (1-\alpha)W_n')^{-1} \vec{1}_n$, where  $I_n$ is the identity matrix of $\mc{H}_n$ and $\vec{\mathbf{1}}_n$ is the $|V_n|\times 1$ vector of all $1$s. Then 
\begin{align}
v_W = S \bigg{(} \frac{|V_1| }{N}  v_1 \oplus  \frac{ |V_2|}{N} v_2 \bigg{)}. \label{eq:influence-vec}
\end{align}
Furthermore, suppose the Neumann series for $S$ converges absolutely in norm. Denote 
\[
W_{\operatorname{int}}' = \begin{bmatrix} 0 & A_1 \\ A_2 & 0 \end{bmatrix}. 
\]
Then
\begin{align}\label{eq:Smatrix}
S = \begin{bmatrix} X & (1-\alpha) (I_{m_1} - (1-\alpha)W_1')^{-1}A_1 Y   \\ (1-\alpha) (I_{m_2} - (1-\alpha)W_2')^{-1} A_2 X & Y  \end{bmatrix}
\end{align}
where $X$ and $Y$ are given by
\[
X =  \bigg{(} I_{m_1} - (1-\alpha)^2 (I_{m_1} - (1-\alpha)W_1')^{-1} A_1  (I_{m_2} - (1-\alpha)W_2')^{-1} A_2 \bigg{)}^{-1} 
\]
as well as
\[
Y =  \bigg{(} I_{m_2} - (1-\alpha)^2 (I_{m_2} - (1-\alpha)W_2')^{-1} A_2  (I_{m_1} - (1-\alpha)W_1')^{-1} A_1 \bigg{)}^{-1}.
\]
\end{lemma} 

\begin{proof}
First, we rewrite
\begin{align}
I - (1-\alpha)W' = ( I - (1-\alpha)W_1' - (1-\alpha)W_2 ' ) ( I - (1-\alpha)L_\perp W_{\operatorname{int}}') 
\end{align}
and note that 
\begin{align}
L_\perp  & = \begin{bmatrix} (I_{m_1} - (1-\alpha)W_1')^{-1} & 0 \\ 0 & (I _{m_2} - (1-\alpha)W_2')^{-1}  \end{bmatrix}  \nonumber \\ 
& =  \begin{bmatrix} (I_{m_1} - (1-\alpha)W_1')^{-1} & 0 \\ 0 & I _{m_2} \end{bmatrix} \begin{bmatrix} I_{m_1}& 0 \\ 0 & (I _{m_2} - (1-\alpha)W_2')^{-1}  \end{bmatrix}.
\end{align}
Hence
\begin{align}
v_W & = \frac{\alpha}{N} ( I - (1-\alpha)W')^{-1} \vec{\mathbf{1}} \nonumber \\
& =  \frac{\alpha}{N} ( I - (1-\alpha)L_\perp W_{\operatorname{int}}') ^{-1} L_\perp \vec{\mathbf{1}} \nonumber \\
& = S \frac{\alpha}{N} \begin{bmatrix}  (I_{m_1} - (1-\alpha)W_1')^{-1} \vec{\mathbf{1}}_{m_1} \\  (I_{m_2} - (1-\alpha)W_2')^{-1} \vec{\mathbf{1}}_{m_2}\end{bmatrix} \nonumber \\
& =  S \bigg{(} \frac{m_1 }{N}  v_1 \oplus  \frac{m_2 }{N} v_2 \bigg{)}, 
\end{align}
proving equation (\ref{eq:influence-vec}). Next, we prove equation (\ref{eq:Smatrix}) formally. By definition, $W_{\operatorname{int}}'$ is block off-diagonal, with parts
\begin{align}
W_{\operatorname{int}}' = \begin{bmatrix}  0 & A_1 \\ A_2 & 0 \end{bmatrix}.
\end{align}
Denoting $L_{W_n} = (I_{m_n} - (1-\alpha)W_n')^{-1}$, 
\begin{align}
 L_\perp W_{\operatorname{int}}' = \begin{bmatrix} 0 & L_{W_1}A_1 \\ L_{W_2}A_2 & 0  \end{bmatrix}. 
\end{align}
Let $\ell_n = L_{W_n} A_n$ and $\lambda = 1-\alpha$. It follows, 
\begin{align}\label{eq:components}
\sum _{k=0}^\infty ( (1-\alpha)L_\perp W_{\operatorname{int}}') ^{k} = I + \begin{bmatrix}  \sum _{k\geq 1} (\la^2  \ell_1 \ell_2 )^k  & (1-\alpha)\ell_1 (I_{m_2}+\sum _{k \geq 1} (\la^2 \ell_2 \ell_1)^k) \\ (1-\alpha)\ell_2(I_{m_1}+ \sum _{k\geq 1} (\la^2  \ell_1 \ell_2 )^k ) &\sum _{k \geq 1} (\la^2 \ell_2 \ell_1)^k  \end{bmatrix}.
\end{align}
Given that the Neumann series for $S$ converges absolutely, it suffices to prove that the components of equation (\ref{eq:components}) also converge absolutely. Since all norms are equivalent in finite dimensional vector spaces, we consider the operator norm. Then
\begin{align}
\bigg|\bigg| \sum _{k \geq M} (\la^2 \ell_1 \ell_2)^k \bigg|\bigg| \leq \norm{  \begin{bmatrix}  \sum _{k \geq M} (\la^2 \ell_1 \ell_2)^k & 0 \\ 0 &  \sum _{k \geq M} (\la^2 \ell_2 \ell_1)^k \end{bmatrix}  } = \bigg|\bigg| \sum _{\substack{k \geq M \\ k \equiv 0 \modu 2}} ((1-\alpha)L_\perp W_{\operatorname{int}}')^k \bigg|\bigg|. \nonumber
\end{align} 
Hence the partial sums of $X$ are Cauchy. The same argument shows that the series for $Y$ also converges absolutely in norm when the Neumann series for $S$ converges absolutely in norm. Then we identify 
\begin{align}
    \sum _{k \geq 0} ((1-\alpha)\ell_1 \ell_2)^k  = (I_{m_1} - \la^2 \ell_1\ell_2)^{-1} \quad \textnormal{and} \quad \sum _{k \geq 0} ((1-\alpha)\ell_2 \ell_1)^k  = (I_{m_2} - \la^2 \ell_2\ell_1)^{-1}
\end{align} 
to get the statement of the lemma.
\end{proof}

\subsection{Application: directed chains}

Now we apply the results above to a special case where $G$ has the structure of a directed chain. %


\begin{definition}\label{def:directedchain}
The partition $\set{V_{r}}$ is a \textit{directed chain} if (1) for all $r$, the nodes of $V_{r}$ only interact with nodes in $V_{r}, V_{r-1}$ or $V_{r+1}$, with the endpoints $V_1,V_M$ only interacting with $V_1,V_2$ and $V_M, V_{M-1}$, respectively, and (2) if there is an edge between $f\in V_{i}$ and $g\in V_{i+1}$, then the direction is from $f$ to $g$.
\end{definition}
Diagrammatically, if $G$ is a network of firms admitting a directed chain partition, and $G_i$ is the graph defined by nodes $V_i$, then $G$ has the form:
\begin{align}\label{eq:graph}
G_1 \to G_2 \to \ldots \to G_{M-1} \to G_M
\end{align}
where $\to$ denotes all of the weighted and directed edges which connect two components $G_{i}, G_{i+1}$ together. In this case, the input-output matrix $W$ decomposes as
\begin{align}\label{eq:wi}
W = \bigoplus _{i=1}^M W_i + \sum _{i=1}^{M-1} W_{\operatorname{int}}^{(i,i+1)}
\end{align}
in the obvious way generalizing equation (\ref{eq:two-case}). The directed chain assumption implies $$(W_{\operatorname{int}}^{(i,i+1)})' = \begin{bmatrix} 0 & A_1^{(i,i+1)} \\ 0 & 0 \end{bmatrix}.$$

We now introduce two new coupling concepts.

\begin{definition}\label{def:weakcoupling}
Say the directed chain $\set{V_i}$ with input-output matrix $W$ is \textit{weakly coupled} if $$\max _{i} \norm{ (I_{m_i} - (1-\alpha)W_i')^{-1} A^{(i,i+1)}_1} < \frac{1}{1-\alpha}.$$
Define the \textit{coupling constant} as 
\begin{align}\label{eq:cc}
    \gamma = (1-\alpha) \max _{i} \norm{ (I_{m_i} - (1-\alpha)W_i')^{-1} A^{(i,i+1)}_1}.
\end{align}
\end{definition}


\noindent \textbf{Example:} The weak coupling condition is satisfied for the model where $V_i = \set{f_i}$, 
\begin{align}
G: \set{ f_1} \to \set{f_2} \to \cdots \to \set{ f_M} ,
\end{align}
and $ 0< \alpha < 1$. To see this, we note $W_i = [0]$  is a $1\times 1$ matrix, and $A_1^{(i,i+1)} = 1$ for all $i$. 


\begin{theorem}\label{thm:finitecorrelation}
Suppose the chain in \eqref{eq:graph} is weakly coupled. Fix $K < M$. For $q \leq K $, let $P_q$ denote the orthogonal projection onto $\mc{H}_{\set{1,\ldots, q}}$. If $W = [w_{ij}]$ and $U = [u_{ij}]$ are input-output matrices such that $w_{ij} \not = u_{ij} $ only if $$f_{j} \in \bigcup_{l \geq K} V_l,$$ then 
\begin{align}
\norm{ P_q ( v_W - v_U)} \leq \frac{ 2 \sqrt{q} | \cup _{i=K}^{M} V_i|}{N}   e^{- \log(1/\gamma) (K-q)}
\end{align}
where $\gamma$ is the coupling constant from \eqref{eq:cc}.
\end{theorem}

\begin{proof}
Denote $m_i = |V_i|$. First, dividing $V$ into $V_1$ and $V_2 \cup \cdots \cup V_M$, Lemma \ref{lem:basecase} gives a sum decomposition of $v_W$,
\begin{align}
v_W & = S^{(1)} \bigg{(} \frac{m_1}{N} v_1 \oplus \frac{N-m_1}{N} v_{2,\ldots, M}  \bigg{)}  
\end{align}
where $v_{2,\ldots, M}$ is the influence vector of $G_2 \cup \cdots \cup G_M$ and $S^{(1)}$ is defined in equation (\ref{eq:Smatrix}). Iterating using partitions $V_i$ and $V_{i+1}\cup V_{i+2} \cup \cdots \cup V_{M}$ yields
\begin{align}\label{eq:induct}
v_W & = S^{(1)} (I_{m_1}\oplus S^{(2)}) \cdots (I_{N-m_{M-1} - m_{M}} \oplus S^{(M-1)}) \bigg{(} \frac{m_1}{N} v_1 \oplus \frac{m_2}{N} v_2 \oplus  \cdots \oplus \frac{m_M}{N} v_M  \bigg{)}  .
\end{align}
Now we consider $S^{(i)} : \mc{H}_{\set{i,\ldots, M}} \to  \mc{H}_{\set{i,\ldots, M}}$. By assumption on the direction of the graph, $S^{(i)}$ has a block matrix decomposition with respect to the direct sum  $\mc{H}_{\set{i,\ldots, M}} = \mc{H}_{i} \oplus \mc{H}_{\set{i+1,\ldots, M}}$,
\begin{align}\label{eq:weak-couple}
S^{(i)} = \begin{bmatrix} I_{m_i} & \hat{S}^{(i)} \\ 0 & I_{\mu_i}  \end{bmatrix}
\end{align}
where $\mu_i = \sum _{j > i} m_j$ and $\hat{S}^{(i)} = (1-\alpha)(I_{m_i} - (1-\alpha)W_i')^{-1}  A^{(i,i+1)}_1 .$ Applying equation (\ref{eq:weak-couple}) to (\ref{eq:induct}) gives
\begin{align}
v_W & = \frac{1}{N} ( \theta_1 \oplus \theta_2 \oplus \cdots \oplus \theta_M    )
\end{align}
where the new coordinates are,
\begin{align}\label{eq:newcoord}
\theta_k  = m_k v_k + m_{k+1}  \hat{S}^{(k)} v_{k+1} + m_{k+2}  \hat{S}^{(k)} \hat{S}^{(k+1)} v_{k+2} + \cdots + m_M \hat{S}^{(k)} \cdots \hat{S}^{(M-1)}  v_M.
\end{align}
Now, we observe that since the $p$-norm is decreasing for $p\geq 1$, 
\begin{align}\label{eq:unitvec}
\norm{v_j} \leq \norm{v_j}_1 \leq 1.
\end{align}
And by the weak coupling assumption, 
\begin{align}\label{eq:normS}
\norm{\hat{S}^{(j)}} \leq (1-\alpha)\max_i \norm{(I_{m_i} - (1-\alpha)W_i')^{-1} A_1^{(i,i+1)}} = \gamma . 
\end{align}
Now, let $\theta^W_i$ and $\theta^U_i$ denote the $i$-th components of $v_W$ and $v_U$, respectively. Suppose $r< K$. Then by the assumption that $U$ and $W$ only differ at the tail of the chain, we see that
\begin{align}
(1-\alpha)(I_{m_r} - (1-\alpha)W_r')^{-1}   = (1-\alpha)(I_{m_r} - (1-\alpha)U_r')^{-1} \textnormal{ and }W^{(r,r+1)} _{\operatorname{int}} =  U^{(r,r+1)}_{\operatorname{int}} .
\end{align}
This and equation (\ref{eq:newcoord}) imply that, for $r<K$, 
\begin{align}
\theta^W _r - \theta^U_r & = m _{K} \bigg{(} \hat{S}^{(r,W)}\cdots \hat{S}^{(K-1,W)} v_{K}^{(W)} - \hat{S}^{(r,U)}\cdots \hat{S}^{(K-1,U)} v_{K}^{(U)}\bigg{)}  +\cdots \nonumber \\ 
 & \hspace{30mm} \cdots +m _M \bigg{(} \hat{S}^{(r,W)}\cdots \hat{S}^{(M-1,W)} v_M^{(W)} - \hat{S}^{(r,U)}\cdots \hat{S}^{(M-1,U)} v_M^{(U)} \bigg{)} \nonumber \\
 &= \sum_{j=K}^M m_j \bigg\{\bigg(\prod_{i=r}^{K-1} {\hat S}^{(i,W)}\bigg)v_K^{(W)} - \bigg(\prod_{i=r}^{K-1} {\hat S}^{(i,U)}\bigg)v_K^{(U)}\bigg\},
\end{align}
where $\hat{S}^{(i,W)}$ and $\hat{S}^{(i,U)}$ are the interaction matrices in (\ref{eq:weak-couple}) for $W$ and $U$, respectively, and $v_i ^{(W)}$ and $v_i^{(U)}$ are the influence vectors for $V_i$ with interactions $W_i$ and $U_i$, respectively. Equations (\ref{eq:unitvec}) and (\ref{eq:normS}) imply 
\begin{align}
\norm{ \theta^W_r - \theta^U_r} \leq  \bigg{(}  \sum _{j=K}^{M} m_j \prod_{i=r}^{j-1} \norm{\hat{S}^{(i,W)} } + \sum _{j=K}^{M} m_j \prod_{i=r}^{j-1} \norm{\hat{S}^{(i,U)} }   \bigg{)}  \leq 2  \gamma^{ K - r} \sum _{j=K} ^{M} m_j . 
\end{align}
It follows that 
\begin{align}
\norm{ P_q ( v_W - v_U)}^{2} & = \frac{1}{N^2} \sum _{i=1}^q \norm{ \theta_i^W - \theta_i^U}^2 \nonumber \\
& \leq  \frac{4}{N^2} \sum _{i=1}^q \gamma ^{2(K-i)}  \bigg(\sum _{j=K}^{M} m_j \bigg) ^2 \nonumber \\
& \leq 4 \bigg{(} \frac{1}{N}\sum _{j=K}^M m_j \bigg{)} ^2 q \gamma^{2(K-q)}.
\end{align}
From this inequality we now obtain
\begin{align}
\norm{ P_q ( v_W - v_U)} & \leq  2\sqrt{q}  \bigg{(} \frac{1}{N} \sum_{j=K}^M m_j\bigg{)} e^{ - \log(\frac{1}{\gamma}  ) (K-q)},
\end{align}
which is the desired bound.
\end{proof}

We now combine this result with Theorem \ref{maxdeltacor} to give a description of how much data is required in a neighborhood of a node in a directed chain of firms to achieve an accurate measurement of systemic importance via the influence vector. For simplicity, we consider neighborhoods of the start of the chain. 


\begin{corollary}\label{cor:combined}
    Suppose that $\set{V_j}_{j=1}^M$ is a weakly coupled directed chain partition of $G$ with input-output matrix $W$, labor constant $\alpha$ and coupling constant $\gamma$. Let $U$ be an observed input-output matrix for $G$ such that for each firm in $V_1 \cup \cdots \cup V_{k+1}$, the missing data accounts for at most a $\delta_k$ share of their intermediate output use. Then, for any $f_i\in V_1$,

    \begin{align}\label{eq:coeff-bound}
        | \ip{ e_i, v_W } - \ip{e_i, v_U }| \leq \frac{(1-\alpha)(2\delta_k - \delta_k^2)}{\alpha(1-\delta_k)} + 2 e^{ - \log(1/\gamma) k} . 
    \end{align}
\end{corollary}

\begin{proof}
    Denote
    \begin{align}
        \tilde{W} = \bigoplus_{i=1}^{k+1} W_i + \sum _{i=1}^{k} W_{\operatorname{int}}^{(i,i+1)} + U_{>k+1}, 
    \end{align}
    where 
    \begin{align}
        U_{>k+1} = \bigoplus _{i=k+2}^{M} U_i + \sum _{i=k+1}^{M-1} U_{\operatorname{int}}^{(i,i+1)}
    \end{align}
    and $W_i, U_i$ and $W_{\operatorname{int}}^{(i,i+1)}, U_{\operatorname{int}}^{(i,i+1)}$ are defined as in equation (\ref{eq:wi}).
    Then, if $P_1$ denotes the orthogonal projection onto $\mc{H}_{ V_1  }$,  
    \begin{align}
        | \ip{ e_i, v_W } - \ip{e_i, v_U }| & \norm{P_1 ( v_W - v_U) } \leq \norm{P_1(v_W - v_{\tilde{W}})} + \norm{ P_1(v_{\tilde{W}} - v_U) }.
    \end{align}
By Theorem \ref{thm:finitecorrelation}, 
\begin{align}\label{eq:lineone}
    \norm{P_1(v_W - v_{\tilde{W}})} \leq 2 e^{ - \log(1/\gamma) k}, 
\end{align}
and by Theorem \ref{maxdeltacor}, 
\begin{align}\label{eq:linetwo}
    \norm{ P_1(v_{\tilde{W}} - v_U) } \leq \frac{(1-\alpha)(2\delta_k - \delta_k^2)}{\alpha(1-\delta_k)}.
\end{align}
Combining (\ref{eq:lineone}) and (\ref{eq:linetwo}) gives the righthand side of (\ref{eq:coeff-bound}). 
\end{proof}

\section{A negative result: degrees of missing information}\label{s:missdegrees}

Given Theorem~\ref{thm:finitecorrelation}, it is natural to investigate how far off the influence vector can be when we only have data on nodes and connections that are within distance $k$ of some source node. In the next result, we construct a network $G$ with a source node $v$ and an observed network $H$ that is missing data on nodes and connections in $G$ which are not within $k$ steps of $v$, for which the $1$-norm of the difference between the influence vectors is $\Omega(1)$ as long as $k = o(n)$. Theorem \ref{thm:counterex-91} shows that some assumption on the structure of graph, such as the conditions in Definitions \ref{def:directedchain} and \ref{def:weakcoupling}, is necessary for locality results like Theorem \ref{thm:finitecorrelation}.  

\begin{theorem}\label{thm:counterex-91}
There exists a network $G$ of order $n > 3$ with a source node $v$ and an observed network $H$ that is obtained from $G$ by eliminating nodes and connections which are not within $k = o(n)$ steps of $v$, for which the $1$-norm of the difference between the influence vectors of $G$ and $H$ is $\Omega(1)$, where the constant in the bound depends on $\alpha$, but not on $n$ or $k$.
\end{theorem}

\begin{proof}
Fix some real number $b > 0$. Consider a network $G$ of order $n$ with vertices $v_1, v_2, \dots, v_n = v$ with $n \times n$ connectivity matrix $W$ which satisfies $w_{1 j} = \frac{1}{n-1}$ for $j > 1$, $w_{i j} = \frac{b}{n-2}$ for $i > 1$ and $j \not \in \left\{i-1,i\right\}$, $w_{i (i-1)} = 1-b$ for $i > 1$, and all other entries are zero. As $b \rightarrow 0$, Theorem~\ref{maxdeltacor} implies that the influence vector of $G$ converges to the influence vector of the network $J$ with connectivity matrix $X$ which satisfies $x_{1 j} = \frac{1}{n-1}$ for $j > 1$, $x_{i (i-1)} = 1$ for $i > 1$, and all other entries are zero. Let $v_X = (p_1, p_2, \dots, p_n)$ be the influence vector of the network $J$. By definition, we have \begin{align} p_n &= \frac{\alpha}{n}+\frac{1-\alpha}{n-1}p_1, \nonumber \\
 p_{n-i} &= \frac{\alpha}{n}+\frac{1-\alpha}{n-1}p_1+(1-\alpha)p_{n-i+1} \text{ for } 0 < i < n-1, \nonumber \\
 p_1 &= \frac{\alpha}{n}+(1-\alpha)p_2.
\end{align} Thus, we have \begin{align}
 p_{n-1}-p_n &= (1-\alpha)p_n, \nonumber \\
 p_{n-i}-p_{n-i+1} &= (1-\alpha)(p_{n-i+1}-p_{n-i+2}) \text{ for } 1 < i < n-1.
\end{align} The last equation is a homogeneous linear recurrence which has the solution \[p_{n-i} = \bigg(\frac{1}{\alpha}-\frac{(1-\alpha)^{i+1}}{\alpha}\bigg)p_n \text{ for } 0 \le i < n-1.\] Thus, \[p_1 = \frac{\alpha}{n}+(1-\alpha)p_2 = \frac{\alpha}{n}+(1-\alpha)\bigg(\frac{1}{\alpha}-\frac{(1-\alpha)^{n-1}}{\alpha}\bigg)p_n.\] Therefore, \[p_n = \frac{(1-\frac{\alpha}{n})\alpha^2}{-\alpha^2+(1-\alpha)^{n+1}+\alpha n + \alpha-1}.\] Note that $p_{n-i}$ is increasing with respect to $i$ for $0 \le i < n-1$, $p_1$ is a weighted average of $p_2$ and $\frac{1}{n}$, and $\sum_{t = 1}^n p_t = 1$. Thus, we must have $p_2 \ge p_1 \ge \frac{1}{n}$ and $p_n \le \frac{1}{n}$. Since $p_t > \frac{\alpha}{n}$ for all $t$ by definition of $G$ and $p_{n-i} \le \frac{p_n}{\alpha}$ for $0 < i < n-1$, we have \[p_t \in \bigg[\frac{\alpha}{n},\frac{1}{\alpha n}\bigg] \text{ for all } t \text{ with } 1 \le t \le n.\]

Now, consider the network $H$ obtained from $G$ by using only the nodes and connections that are within distance $k$ of vertex $v = v_n$. In particular, $H$ has the $(k+1) \times (k+1)$ connectivity matrix $U$ with rows and columns indexed by $n-k, n-k+1, \dots, n$ for which $u_{(n-k)j} = \frac{1}{k}$ for $j > n-k$, $u_{i j} = \frac{b}{(n-2)(1-b+\frac{(k-1)b}{n-2})}$ for $i > n-k$ and $j \not \in \left\{i-1,i\right\}$, $u_{i (i-1)} = \frac{1-b}{1-b+\frac{(k-1)b}{n-2}}$ for $i > n-k$, and all other entries are zero. As $b \rightarrow 0$, Theorem~\ref{maxdeltacor} implies that the influence vector of $H$ converges to the influence vector of the network $K$ which has connectivity matrix $Z$ with rows and columns indexed by $n-k, n-k+1, \dots, n$ and satisfies $z_{(n-k) j} = \frac{1}{k}$ for all $j > n-k$, $z_{i (i-1)} = 1$ for $i > n-k$, and all other entries are zero. 

Note that $Z$ is the same as $X$, except that every $n$ has been replaced with $k+1$ and the row and column indices start from $n-k$ in $Z$ rather than $1$. Thus, if $v_Z = (q_{n-k}, q_{n-k+1}, \dots, q_{n})$ denotes the influence vector of $Z$, then we have \[q_t \in \bigg[\frac{\alpha}{k+1},\frac{1}{\alpha (k+1)}\bigg] \text{ for all } t \text{ with } n-k \le t \le n.\]

In order to compare $v_X$ and $v_Z$, we must prepend $n-k-1$ zeroes to $v_Z$ so that both vectors have the same length. The first $n-k-1$ coordinates each differ by at least $\frac{\alpha}{n}$ between $v_X$ and $v_Z$, while the last $k+1$ coordinates each differ by at least $\frac{\alpha}{k+1}-\frac{1}{\alpha n}$ between $v_X$ and $v_Z$. Therefore, we have \[\norm{v_X-v_Z}_1 \ge (n-k-1)\frac{\alpha}{n}+(k+1)\bigg(\frac{\alpha}{k+1}-\frac{1}{\alpha n}\bigg) = \Omega(1) \text{ for } k = o(n),\] where the constants in the bound depend on $\alpha$, but not on $n$ or $k$. Thus, as $b \rightarrow 0$, we have $\norm{v_W-v_U}_1 \rightarrow \norm{v_X-v_Z}_1 = \Omega(1)$ for $k = o(n)$.
\end{proof}

\section{Implications for \texttt{PageRank}}\label{s:pagerank}

In addition to networks of input-output linkages between firms, the results in this paper can also be applied to networks of linkages between websites or citation networks. Given some collection of $n$ websites $s_1, \dots, s_n$, let $W$ be a weighted linkage network where the vertex set consists of websites. With probability $1-\alpha$, if the visitor is at website $s_i$, then they randomly select one of the websites linked from $s_i$ to visit next if $s_i$ has links, and they choose a new website uniformly at random if $s_i$ has no links. With probability $\alpha$, the visitor uniformly at random selects any website in the network to visit next, including staying at $s_i$. Specifically, in the $\alpha$ case the visitor goes to website $s_j$ with probability $\frac{1}{n}$ for each website $s_j$ in the network. Note that this $\alpha$ case prevents the random visitor from getting stuck at a website if it has no links. For this $\alpha$ case, we can think of the visitor typing in a url for the next website rather than clicking.

For website $s_i$, let $\outdeg(s_i)$ denote the number of websites to which website $s_i$ has a link. With probability $1-\alpha$, the visitor goes to some website which is linked from their current website, assuming that their current website has at least one link. Specifically, the visitor goes to website $s_j$ with probability $w_{i j}$, when $s_j$ is some website that can be clicked from website $s_i$. Here we take $w_{i j} = \frac{1}{\outdeg(s_i)}$ if $\outdeg(s_i) > 0$, but we could also let $w_{i j}$ vary with respect to the prominence of the link to $s_j$ on website $s_i$. When $\outdeg(s_i) = 0$, we let $w_{i j} = \frac{1}{n-1}$ for every $j \neq i$, i.e., the visitor goes to any website besides $s_i$ uniformly at random. For the two cases with probabilities $\alpha$ and $1-\alpha$, we can think of the visitor having a biased coin where one side lands with probability $\alpha$ and the other side lands with probability $1-\alpha$. We call $\alpha$ the \text{bias parameter}.

The \texttt{PageRank} vector $q_W$ corresponding to the network with linkage matrix $W$ and parameter $\alpha$ is defined the same as the influence vector \cite{pagerank}, i.e. $q_W = v_W$. The following results for \texttt{PageRank} are immediate implications of our main theorems for the influence vector. The first follows from Theorem~\ref{thmsharpdelta} and the second follows from Corollary~\ref{cor:combined}.

\begin{theorem}
Let $W$ be the true matrix of linkages with all of the data, and let $U$ be the observed matrix of the input-output linkages with data missing. Suppose that for each $i = 1, \dots, n$, the missing data accounts for at most a $\delta$ share of the links from each website $s_i$, where $0 \le \delta < 1$. In other words, at most $\delta \outdeg(s_i)$ of the links from $s_i$ are not crawled for each $i$. If $q_U$ denotes the \texttt{PageRank} vector calculated with the observed data and $q_W$ denotes the true \texttt{PageRank} vector, then the maximum possible value of $\norm{q_U-q_W}_p$ is $\Theta(\delta)$ for all $p \ge 1$, where the constants in the bound depend on $\alpha$.
\end{theorem}

\begin{theorem}
Suppose that $\set{V_j}_{j=1}^M$ is a weakly coupled directed chain partition of the website network $G$ with linkage matrix $W$, bias parameter $\alpha$ and coupling constant $\gamma$. Let $U$ be an observed linkage matrix for $G$ such that for each website $s_i$ in $V_1 \cup \cdots \cup V_{k+1}$, the missing data accounts for at most a $\delta_k$ share of the links from website $s_i$ (i.e. at most $\delta_k \outdeg(s_i)$ of the links from $s_i$ are not crawled). Then, for any $s_i\in V_1$,

    \begin{align}
        | \ip{ e_i, q_W } - \ip{e_i, q_U }| \leq \frac{(1-\alpha)(2\delta_k - \delta_k^2)}{\alpha(1-\delta_k)} + 2 e^{ - \log(1/\gamma) k} . 
    \end{align}
\end{theorem}

The following table illustrates analogues between \texttt{PageRank} and IO models.

\begin{center}
\begin{longtable}
{| p{3.0cm} || p{5.4cm} | p{7.0cm} |}
\hline
\textbf{Concept} & \textbf{\texttt{PageRank}} & \textbf{Input-Output} \\ 
\hline 
Vertices & Websites & Firms \\
\hline
Directed Edges & Links & Contributions between firms \\
\hline
Random walker & Stepping between websites & Stepping between firms \\
\hline
Step (bias $1-\alpha$) & Click link if one exists & Go to firm that contributes intermediate input if one exists \\
\hline
Teleport (bias $\alpha$) & Type \texttt{url} & Go to firm through laborer that uses all firms uniformly \\
\hline
Matrix $W$ & $w_{i j} = \frac{1}{\outdeg(\text{website i})}$ if website $i$ has link to website $j$, $0$ otherwise & $w_{i j} = $ share of firm $j$ in the intermediate input use of firm $i$ \\
\hline
Leontief inverse & Traffic propagation & Downstream propagation to other firms \\
\hline
Influence vector & Probability of being at website $i$ after many rounds & Probability of being at firm $i$ after many rounds \\
\hline
Influence vector error bound & Maximum possible error with $\delta$ fraction of links not crawled for website $i$ & Maximum possible error with $\delta$ share of data missing for intermediate input use of firm $i$ \\
\hline
Directed chain & Chain of linked websites & Supply chain \\
\hline
Weak coupling & Connectivity of adjacent linked groups of websites & Connectivity of adjacent dependent groups of firms \\
\hline
\end{longtable}
\end{center}

\section{Future work}\label{s:future}
Our underlying economic model has been the Cobb-Douglas model. While this has been a very useful model, it lacks certain flexibility since we are assuming that an industry's expenditure on various inputs as a fraction of its sales is invariant to the realization of shocks. We can alleviate this situation by relaxing this invariance and not having it hold for more general production technologies. This situation is much richer and also closer to real world scenarios, see e.g. [$\mathsection$ 2.3.1]\cite{premier}. 

One of the strengths of the approach outlined in this work is that the influence vector is independent of prices and is only concerned with changes in quantities and how they percolate through the economic network defined by the input-output matrix. As we move to more flexible functional forms for the production function, output prices become increasingly important and the influence vector may no longer be handled through linear algebra and may require more complicated polynomial expansions as approximations. The reality is that demand functions are likely to not have clean closed form solutions that can readily be manipulated in the way that we have done with the Leontief inverse matrix.

One such generalization is called the Constant Elasticity of Substitution (CES) model, see e.g. \cite{McFadden} and \cite{ArrowETal}, and it can be described as follows. Let us suppose that the production technology of firms in industry $i$ is now given by
\begin{align} \label{eq:ces}
    y_i = z_i \zeta_i l_i^{\alpha_i} \bigg(\sum_{j=1}^n a_{ij}^{1/\sigma_i} x_{ij}^{1-1/\sigma_i}\bigg)^{(1-\alpha_i)\sigma_i/(\sigma_i-1)},
\end{align}
where, once again, \eqref{eq:constantreturn} holds but now $\sigma_i$ denotes the elasticity of substitution between the multiple inputs. Unlike in the Cobb-Douglas model, the normalization constant is now chosen to be $\zeta_i = \alpha_i^{-\alpha_i} (1-\alpha_i)^{-(1-\alpha_i)\sigma_i/(\sigma_i-1)}$. This situation is clearly more general as we have an \textit{individual} elasticity $\sigma_i$ associated to the inputs. In order to recover \eqref{eq:cobbdouglas} from \eqref{eq:ces}, we need a careful singularity analysis as letting $\sigma_i \to 1$ for all $i$ will produce an undesirable $0/0$ situation. Indeed focusing, for simplicity, on $n=2$ yields
\begin{align} 
    y_i = z_i \zeta_i l_i^{\alpha_i} [a_{i1}^{1/\sigma_i} x_{i1}^{1-1/\sigma_i} + a_{i2}^{1/\sigma_i} x_{i2}^{1-1/\sigma_i}] ^{(1-\alpha_i)\sigma_i/(\sigma_i-1)}.
\end{align}
Here $a_{i2}=1-a_{i1}$. Taking logarithms on both sides yields
\begin{align} 
    \log y_i = \log(z_i \zeta_i \ell_i^{\alpha_i}) + \frac{(1-\alpha_i)\sigma_i}{\sigma_i-1} \log [a_{i1}^{1/\sigma_i} x_{i1}^{1-1/\sigma_i} + (1-a_{i1})^{1/\sigma_i} x_{i2}^{1-1/\sigma_i}].
\end{align}
If $\sigma_i = 1$, then the denominator in the fraction of the second term will tend to zero and the last term involving the logarithm will be of the form $\log[a_{i1}+1-a_{i1}]=\log 1 = 0$. Therefore, we invoke l'H\^{o}pital's rule so that as $\sigma_i \to 1$ we have
\begin{align}
    \log y_i = \log(z_i \zeta_i \ell_i^{\alpha_i}) + a_{i1}\log x_{i1} + (1-a_{i1})\log x_{i2} = \log(z_i \zeta_i \ell_i^{\alpha_i}) + a_{i1}\log x_{i1} + a_{i2}\log x_{i2}. \nonumber
\end{align}
Exponentiation on both sides yields
\begin{align}
    y_i 
    = z_i \zeta_i \ell_i^{\alpha_i} \prod_{j=1}^2 x_{ij}^{a_{ij}}. \nonumber
\end{align}
An induction argument can generalize the above to general $n \ge 1$ therefore yielding the baseline model \eqref{eq:cobbdouglas} as the limiting case when $\sigma_i \to 1$ for all $i$. 

In this situation, producing the explicit formula for the effect of a shock to industry $j$ on the output of industry $i$ is much more difficult. In the supplementary materials of \cite{premier}, Carvalho and Tahbaz-Salehi show that up to first order the effect of the shock is given by
\[
\frac{d\log y_i}{d \log z_i}\bigg|_{\log z = 0} = \ell_{ij} + \frac{1}{\lambda_i} \sum_{k=1}^n (\sigma_k-1)\lambda_k \bigg\{\sum_{r=1}^n a_{kr} \ell_{ri} \ell_{rj} - \frac{1}{1-\alpha_k}\bigg(\sum_{r=1}^n a_{kr}\ell_{ri}\bigg)\bigg(\sum_{r=1}^n a_{kr}\ell_{rj}\bigg)\bigg\}.
\]
The economy's Leontief inverse remains $L_A = (I-(1-\alpha)A)^{-1} = [[\ell_{ij}]]_{n \times n}$ and $\lambda_i$ denotes industry $i$'s Domar weight given by \eqref{eq:domar}. This expression admits a very useful interpretation:
\begin{enumerate}
    \item The first term is congruent with the baseline Cobb-Douglas model and it represents the downstream output effect we were seeing in the baseline model.
    \item The second term captures the reallocation effect.
\end{enumerate}
We can now see that the zeroth order model corresponds to the Cobb-Douglas models and that higher order terms will yield more accuracy and they will be interpreted as corrections to the baseline model. 
One direction for future work is to investigate more sophisticated measures of importance which pertain to the CES model. This would require a more careful analysis of the Leontief inverse and how it interacts with the parameters of the model.

\begin{center}
\begin{longtable}
{| p{3.0cm} || p{5.0cm} | p{7.4cm} |}
\hline
\textbf{Model name} & \textbf{Cobb-Douglas} & \textbf{Constant Elasticity of Substitution} \\ 
\hline 
Industry $i$ ouput & $y_i = z_i \zeta_i \ell_i^{\alpha_i} \prod_{j=1}^n x_{ij}^{a_{ij}}$ & $y_i = z_i \zeta_i l_i^{\alpha_i} \big(\sum_{j=1}^n a_{ij}^{1/\sigma_i} x_{ij}^{1-1/\sigma_i}\big)^{(1-\alpha_i)\sigma_i/(\sigma_i-1)}$ \\
\hline
Elasticity & $\sigma=1$ & $\{\sigma_i\}_{i=1}^n$ with $0 < \sigma_i < 1$\\
\hline
Limiting case & Not applicable & Letting $\sigma_i \to 1$ for all $i$ yields Cobb-Douglas. \\
\hline
Leontief inverse & $L_W = (I-(1-\alpha)W')^{-1} = [[\ell_{ij}]]_{1 \le i,j \le n}$ & $L_W = (I-(1-\alpha)W')^{-1} = [[\ell_{ij}]]_{1 \le i,j \le n}$ \\
\hline
Influence vector & ${\vec v}_W = \frac{\alpha}{n}[I-(1-\alpha)W']^{-1} \mathbf{1} $ & ${\vec v}_W = \frac{\alpha}{n}[I-(1-\alpha)W']^{-1} \mathbf{1} $ \\
\hline
Error bound\footnote{This is the worst-case error bound when $U$ is missing at most $\delta$ share of data for intermediate input use of each firm.} & $\norm{\vec v_U - \vec v_W}_p = \Theta(\delta)$ & $\norm{\vec v_U - \vec v_W}_p = \Theta(\delta)$ \\
\hline
\end{longtable}
\end{center}

In Section \ref{s:binomial}, we assumed that $y_{ij}= w_{ij} \text{II}_i \ge M$ for all $i,j$ for some universal constant $M$ (e.g. $M = \min_{i,j} \{ y_{ij} \}$). It would be interesting to explore the possibility of relaxing this condition to $\text{II}_i \ge M$ for another universal constant $M$ and obtaining stronger bounds. However, this may also involve a more complicated analysis. The choice of having the missing data be binomially distributed was natural as it is a flexible and fairly realistic distribution. Indeed, under certain conditions, in the limit as $y_{ij}$ tends to infinity, the binomial distribution can be used to approximate the normal distribution, reinforcing the choice of the binomial distribution for missing data. The key technology for obtaining Theorem \ref{thm:binomial} was Chernoff's concentration inequality. On the other hand, there are other equally reasonable choices depending on the economic landscape where the most realistic distribution might not necessarily be symmetric nor supported on the real line. In order to obtain satisfactory bounds in these situations, one may need to employ other concentration inequalities. 

\section{Acknowledgements}\label{s:acknowledgements}
The authors would like to thank Illya Hicks, Brandon Alston, Shannon Prier, and John Bordeaux for their helpful comments, support, and fruitful discussions. This paper was produced using support from the Homeland Security Research Division in extension of prior work supported by the Cybersecurity and Infrastructure Security Agency.

\end{document}